\documentclass{article}

\usepackage{microtype}
\usepackage{graphicx}
\usepackage{subcaption}
\usepackage{booktabs} 
\usepackage[table]{xcolor}

\usepackage{hyperref}

\usepackage[linesnumbered,ruled,lined]{algorithm2e}

\usepackage[accepted]{icml2026}

\usepackage{amsmath}
\usepackage{amssymb}
\usepackage{mathtools}
\usepackage{amsthm}
\usepackage{xspace}  
\usepackage{multirow}
\usepackage{tcolorbox}
\usepackage{enumitem}
\usepackage{makecell}

\usepackage[capitalize,noabbrev]{cleveref}

\theoremstyle{plain}
\newtheorem{theorem}{Theorem}[section]
\newtheorem{proposition}[theorem]{Proposition}

\newtheorem{corollary}[theorem]{Corollary}
\theoremstyle{definition}

\newtheorem{assumption}[theorem]{Assumption}
\theoremstyle{remark}
\newtheorem{remark}[theorem]{Remark}

\usepackage[textsize=tiny]{todonotes}

\newcommand{\ourmetric}{\textsc{LM-CC}\xspace}

\icmltitlerunning{Rethinking Code Complexity Through the Lens of Large Language Models}

\begin{document}

\twocolumn[
  \icmltitle{Rethinking Code Complexity Through the Lens of Large Language Models}



  \icmlsetsymbol{equal}{*}

  \begin{icmlauthorlist}
    \icmlauthor{Chen Xie}{sjtu}
    \icmlauthor{Xiaodong Gu}{sjtu}
    \icmlauthor{Yuling Shi}{sjtu}
    \icmlauthor{Beijun Shen}{sjtu}
  \end{icmlauthorlist}

  \icmlaffiliation{sjtu}{School of Computer Science, Shanghai Jiao Tong University, Shanghai, China}

  \icmlcorrespondingauthor{Beijun Shen}{bjshen@sjtu.edu.cn}

  \icmlkeywords{Code Complexity Metrics, Large Language Models, Machine Perception of Code, Semantic Compositional Hierarchy}

  \vskip 0.3in
]



\printAffiliationsAndNotice{}  

\begin{abstract}
Code complexity metrics such as cyclomatic complexity have long been used to assess software quality and maintainability. With the rapid advancement of large language models (LLMs) on coding tasks, an important yet underexplored question arises: \textit{do traditional complexity metrics meaningfully characterize the coding difficulty that LLMs perceive?} 
In this work, we empirically demonstrate that classical complexity metrics exhibit no consistent correlation with LLM performance, revealing a fundamental mismatch with model-perceived difficulty.
To address this gap, we propose \ourmetric, a novel code complexity metric tailored for LLMs, grounded in the hypothesis that model-perceived code difficulty is fundamentally driven by semantic nonlinearity.
\ourmetric quantifies complexity through an entropy-guided semantic compositional hierarchy, capturing the cumulative uncertainty encountered by LLMs during code understanding.
Our experimental results demonstrate that \ourmetric exhibits strong and consistent partial correlations with LLM performance, while semantics-preserving reductions in \ourmetric consistently lead to improved downstream task performance.
The source code is available at: \url{https://github.com/xchen121/lm-cc}.
\end{abstract}

\section{Introduction}

\begin{figure}[t]
    \centering
    \includegraphics[width=\columnwidth]{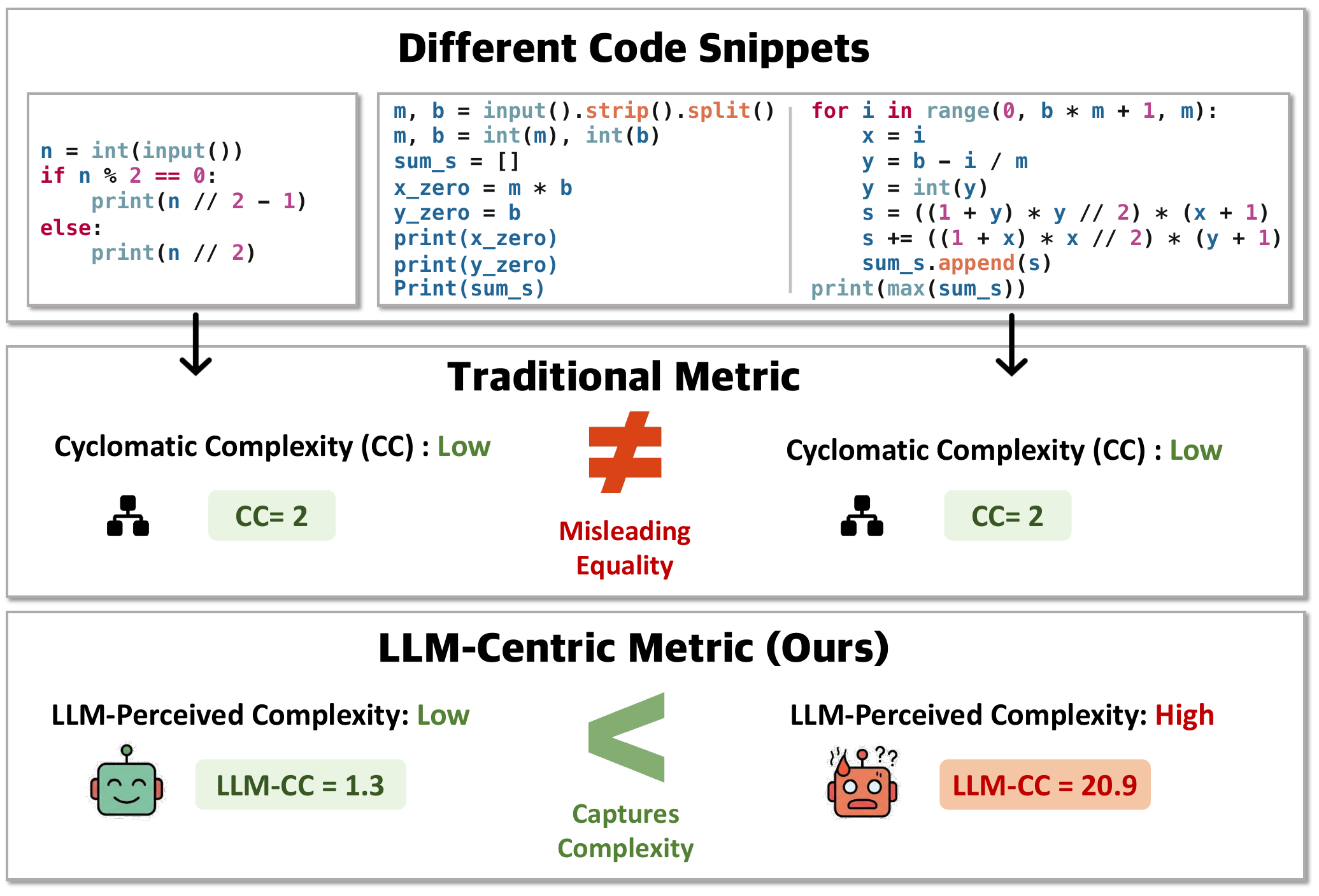}
    \caption{Comparison between Cyclomatic Complexity (CC) and our proposed \ourmetric. While CC assigns identical values to code snippets with significantly different cognitive loads for LLMs (top), \ourmetric effectively distinguishes them by capturing the model uncertainty on nonlinear code semantics (bottom).}
    \label{fig:motivation}
\end{figure}

Understanding and quantifying code complexity has long been a fundamental problem in software engineering~\citep{Weyuker88, AjamiWF19, Feitelson23}. Traditional complexity metrics, such as cyclomatic complexity~\citep{mccabe1976complexity}, were originally designed to characterize structural properties of programs, including control-flow diversity, nesting depth, and logical branching. These metrics serve as quantitative proxies for the human effort required to comprehend, test, and maintain software systems. Consequently, they have been extensively adopted in a wide range of software engineering applications, including refactoring guidance~\citep{GirjoabaC24}, maintainability assessment~\citep{BenarochL23}, and defect prediction~\citep{AlqadiM20}.

Recently, LLMs have achieved remarkable progress across diverse code intelligence tasks, including code completion, generation, summarization, translation, and reasoning~\citep{jiang2025survey,chen2021evaluating,shi2024between}. 
Despite these advances, the notion of code complexity itself has rarely been reconsidered from the perspective of LLMs. 
Existing code complexity metrics primarily characterize program structure from an execution-centric perspective, while largely abstracting away from the representational and inferential mechanisms through which LLMs process code. However, LLMs interpret programs not as executable structures, but as token sequences embedded in high-dimensional semantic spaces. 
For example, code snippets with identical cyclomatic complexity can elicit markedly different performance from LLMs (\cref{fig:motivation}), indicating a potential mismatch between measured complexity and model-perceived difficulty.
This discrepancy raises a critical and underexplored question: \emph{do traditional code complexity metrics meaningfully capture the difficulty that LLMs experience when processing code?}

Answering this question is nontrivial. LLMs operate on code as linear token sequences encoded in distributed representations, and their difficulty arises from structural nonlinearity across composite and branching semantic units, thereby degrading probabilistic inference~\citep{du2025context, shi2025longcodezip, LeVPLNB25}. Consequently, the aspects of code that challenge LLMs may diverge substantially from those emphasized by traditional complexity measures. A systematic investigation of this potential misalignment is therefore essential for accurately characterizing LLM-perceived code complexity and for informing the design of more effective LLM-based code intelligence methods.

In this paper, we conduct an in-depth empirical study on the relationship between traditional code complexity metrics and LLM performance across multiple code-related tasks. To isolate structural complexity from trivial confounding factors, we perform partial correlation analysis with code length controlled as a covariate. Our results show that commonly used complexity metrics exhibit no statistically reliable correlation with LLM performance once length effects are removed, revealing a fundamental gap between classical complexity measures and LLM-perceived difficulty.

Motivated by this observation, we introduce \ourmetric, a novel metric designed to quantify \emph{LLM-perceived Code Complexity}. Our central hypothesis is that the primary source of difficulty for LLMs in code understanding arises from the \emph{nonlinearity of semantic units} inherent in programs. 
Unlike natural language, which is typically interpreted in a largely sequential fashion, source code exhibits intrinsically nonlinear semantic structures~\citep{anand2024critical}. Branching, nesting, and scoping introduce discontinuities in semantic flow, requiring models to reason over multiple hierarchically organized semantic paths rather than a single linear context.

Importantly, LLM difficulty is not determined merely by the presence of nonlinear constructs, but by how such constructs compose and interact. While isolated control-flow elements may pose limited challenges, deeply nested or highly branching structures can substantially increase model uncertainty. This suggests that LLM-perceived complexity depends jointly on compositional level and the density of decision points introduced through structural branching.

To capture these properties, we represent source code using a structured semantic abstraction. Specifically, we decompose programs into semantic units based on entropy (i.e., the surprise that LLMs receive during reading) and organize them into a compositional hierarchy that reflects their jumps and dependencies during LLM processing.
Based on this representation, \ourmetric measures two complementary dimensions of semantic nonlinearity: \emph{compositional level}, reflecting hierarchical nesting, and \emph{branching factor}, reflecting structural fan-out. 
Each semantic unit is assigned a local complexity score derived from its compositional level and derived sub-units, and these scores are recursively aggregated across the hierarchy. 
Intuitively, deeper semantic hierarchies demand more complex multi-level reasoning, while broader branching structures introduce more frequent points of semantic divergence.

We evaluate \ourmetric through extensive experiments across diverse code understanding and generation tasks. The results show that \ourmetric exhibits strong and consistent partial correlations with LLM performance (Spearman’s $r$ ranging from $-0.92$ to $-0.97$) across multiple datasets after controlling for code length, substantially outperforming traditional metrics. 
Moreover, reducing \ourmetric through semantics-preserving code rewriting consistently improves downstream task performance, with gains of up to 20.9\%, demonstrating its practical utility as an effective optimization target.

In summary, this paper makes the following contributions:
\begin{itemize}[itemsep=1pt, topsep=2pt]
    \item We present the first empirical study of code complexity metrics for LLMs, showing that existing metrics fail to reliably capture model-perceived difficulty.
    
    
    \item We propose \ourmetric, a novel code complexity metric explicitly grounded in the perspective of large language models. 
    
    \item We demonstrate the effectiveness and practical utility of \ourmetric through extensive experiments on diverse code-related tasks. 
\end{itemize}

\section{Revisiting Current Complexity Metrics}
\label{sec:existing-metrics}

We begin by systematically re-examining existing code complexity metrics in the context of LLM-based code tasks. Specifically, we investigate whether these metrics correlate with model performance and thereby reflect model-perceived code difficulty.

\subsection{Experimental Setups}
\textbf{Datasets and Tasks}.
We conduct experiments on two widely used benchmark datasets covering three representative code-related tasks: program repair, code translation, and code execution reasoning.
xCodeEval~\citep{khan2024xcodeeval} is a large-scale multilingual and multitask benchmark for code understanding and generation. 
From this benchmark, we consider two tasks:
(1) \emph{program repair} (xCodeEval-apr), where the model is given buggy code and is required to generate a corrected implementation that passes all test cases; and
(2) \emph{code translation} (xCodeEval-ct), where the model translates Python programs into C.
Following prior work, we use filtered subsets consisting of 271 problems for program repair and 535 problems for code translation.
We additionally evaluate on \textsc{HumanEval}~\citep{chen2021evaluating} under the \emph{code execution reasoning} setting~\citep{liu2026codemind}. In this task, the model is required to infer the correct program output given a code snippet and its corresponding input. The evaluation set contains 162 manually curated programming problems.




\textbf{Model Configurations.}
We use \textsc{DeepSeek-V3}\footnote{https://api.deepseek.com} for all experiments, with task-specific hyperparameters configured according to the settings used in the original papers.

\textbf{Evaluation Metric.}
We adopt \emph{pass@1} as the evaluation metric for all tasks.
Pass@1 measures the probability that a single model-generated solution correctly solves a given problem by passing all associated test cases.

\textbf{Correlation Analysis.}
To assess the effectiveness of existing complexity metrics, we analyze their correlation with LLM performance across tasks.
Following~\citet{liu2026codemind}, we perform subgroup-based correlation analysis by binning samples into approximately ten equal-sized groups and computing group-level Spearman correlations using median feature values and average task scores.
Further details are provided in Appendix~\ref{sec:correlation-analysis}.

We report both zero-order and partial correlations. The zero-order correlation measures the monotonic association between two variables without controlling for confounding factors, and is computed as the Pearson correlation over ranked observations:
\begin{equation}
    r_{xy} = \frac{\sum_{i=1}^n (R_i - \bar{R})(S_i - \bar{S})}{\sqrt{\sum_{i=1}^n (R_i - \bar{R})^2} \sqrt{\sum_{i=1}^n (S_i - \bar{S})^2}},
    \end{equation}
where $R_i = R(x_i)$ and $S_i = R(y_i)$ denote the ranks of variables $x$ and $y$, and $\bar{R}$ and $\bar{S}$ are their mean ranks.
To account for the confounding effect of code length, we further compute the partial correlation coefficient $r_{xy \cdot z}$, which quantifies the association between $x$ and $y$ while controlling for $z$:
\begin{equation}
r_{xy \cdot z} = \frac{r_{xy} - r_{xz} \cdot r_{yz}}{\sqrt{1-r_{xz}^2} \cdot \sqrt{1-r_{yz}^2}}.
\end{equation}

\textbf{Complexity Metrics}. We assess the following representative code complexity metrics: 
\begin{itemize}[itemsep=0pt, topsep=1pt]
    \item \textbf{Cyclomatic Complexity (CC)}~\citep{mccabe1976complexity}: quantifies the number of linearly independent execution paths in a program by analyzing its control-flow graph.
    \item \textbf{Halstead Complexity (HC)}~\citep{1977Elements}: measures program complexity via the distribution of operators and operands. 
    Among its derived measures, we focus on \emph{Halstead Difficulty}, which serves as a proxy for program comprehension effort.
    \item \textbf{Maintainability Index (MI)}~\citep{OMAN1994251}: combines Halstead Volume, cyclomatic complexity, lines of code, and comment percentage into a single predictive measure of software maintainability.
    \item \textbf{Cognitive Complexity (CoC)}~\citep{campbell2018cognitive}: estimates code understandability by penalizing nested control structures and nonlinear control flow.    
\end{itemize}

\subsection{Preliminary Results}
\label{sec:preliminary-results}

\begin{table}[t]
    \centering
    \caption{Correlation between traditional code complexity metrics and pass@1 for code tasks using \textsc{DeepSeek-V3}. We report both zero-order correlation and partial correlation while controlling for code length. ``--'' indicates non-significant results ($p \geq 0.05$).}
    \footnotesize
    \label{tab:primary_results}
    \begin{tabular}{l l cc}
        \toprule
        \multirow{2}{*}{\textbf{Metric}} 
        & \multirow{2}{*}{\textbf{Code Task}} 
        & \multicolumn{2}{c}{\textbf{Correlation}} \\
        \cmidrule{3-4}
        & & \textbf{Zero-order} & \textbf{Partial} \\
        \midrule

        \multirow{3}{*}{CC}
            & Program Repair        & - & - \\
            & Code Translation      & - & -0.73 \\
            & Execution Reasoning   & -0.78 & -0.67 \\

        \midrule
        \multirow{3}{*}{HC}
            & Program Repair        & -0.97 & -0.93 \\
            & Code Translation      & -0.90 & -0.70 \\
            & Execution Reasoning   & - & - \\

        \midrule
        \multirow{3}{*}{MI}
            & Program Repair        & - & - \\
            & Code Translation      & - & - \\
            & Execution Reasoning   & - & -0.72 \\
            
        \midrule
        \multirow{3}{*}{CoC}
            & Program Repair        & -0.82 & - \\
            & Code Translation      & - & - \\
            & Execution Reasoning   & - & - \\

        \bottomrule
    \end{tabular}
    \vspace{-8pt}
    \end{table}

Table~\ref{tab:primary_results} summarizes the relationship between traditional code complexity metrics and LLM performance, reporting both zero-order correlations and length-controlled partial correlations. An en dash (``--'') denotes statistically non-significant results.
Across all three tasks, none of the existing metrics exhibits stable or consistently significant associations with LLM performance after controlling for code length.

\emph{Maintainability Index (MI)} achieves a moderate partial correlation ($r=-0.72$) only for the code execution reasoning task, while \emph{Cognitive Complexity (CoC)} exhibits a relatively strong zero-order correlation ($r=-0.79$) in code translation. However, neither metric demonstrates statistically reliable correlations across the remaining tasks, suggesting limited generalizability to model-perceived code difficulty.

Similarly, \emph{Halstead Complexity (HC)} shows strong zero-order correlations in program repair and code translation ($-0.97$ and $-0.90$, respectively), while \emph{Cyclomatic Complexity (CC)} attains a zero-order correlation of $-0.78$ in execution reasoning. Nevertheless, these correlations substantially weaken once code length is controlled for through partial correlation analysis. Moreover, neither HC nor CC yields consistently significant partial correlations across tasks, indicating that these metrics fail to robustly capture the structural characteristics that challenge LLM reasoning.





The systematic disappearance of correlations under length control is not incidental.
Rather, it reflects a fundamental limitation shared by existing metrics.
These metrics are grounded in assumptions about explicit control-flow analysis, human cognitive load, and long-term maintainability, whereas LLMs process programs as token sequences encoded in learned representations.
For LLMs, difficulty arises from nonlinear semantic composition and long-range hierarchical dependencies that are not explicitly modeled by traditional complexity measures.
This empirical decoupling highlights the need for a fundamentally different notion of code complexity, explicitly aligned with the representational and computational characteristics of large language models.

\section{Proposed Metric: \ourmetric}
\label{sec:new-metric}

\subsection{Overview}

Motivated by the limitations of existing metrics, we introduce \ourmetric, a principled measure of code complexity tailored to large language models. By organizing entropy-aligned semantic units into a hierarchical representation, \ourmetric goes beyond surface syntactic abstractions and explicitly encodes the semantic composition and branching relations that govern program semantics and shape LLM uncertainty. While isolated nonlinear constructs are often tractable in isolation, their hierarchical composition induces long-range semantic dependencies and parallel execution contexts, which compound and substantially amplify uncertainty during inference. Consequently, LLM-perceived code complexity cannot be adequately captured by any single structural statistic, but instead arises from the joint interaction between the depth of semantic hierarchies and the density of structurally induced decision points.

Guided by this formulation, we construct \ourmetric through a three-stage framework. 
First, we decompose programs into semantic units via a hybrid entropy–structure mechanism, and organize them into a semantic compositional hierarchy that explicitly encodes nonlinear program semantics (\S~\ref{sec:hierarchy}). 
Second, we extract a set of candidate structural features from this representation—such as compositional level and branching factor—and identify those most predictive of LLM performance through empirical analysis (\S~\ref{sec:features}). 
Finally, we formalize \ourmetric as a weighted aggregation of the selected features (\S~\ref{sec:metric}).
The complete algorithmic procedure is provided in Appendix~\ref{sec:algorithm}.

Beyond the empirical formulation, we further develop a theoretical interpretation connecting the structure encoded by \ourmetric to the inference dynamics of autoregressive language models. This analysis provides insight into why \ourmetric more effectively characterizes LLM-perceived code difficulty than conventional complexity measures (\S~\ref{sec:theory-main}). 

\subsection{Hierarchical Semantic Decomposition}
\label{sec:hierarchy}

\begin{figure*}[t]
    \centering
    \includegraphics[width=0.9\textwidth]{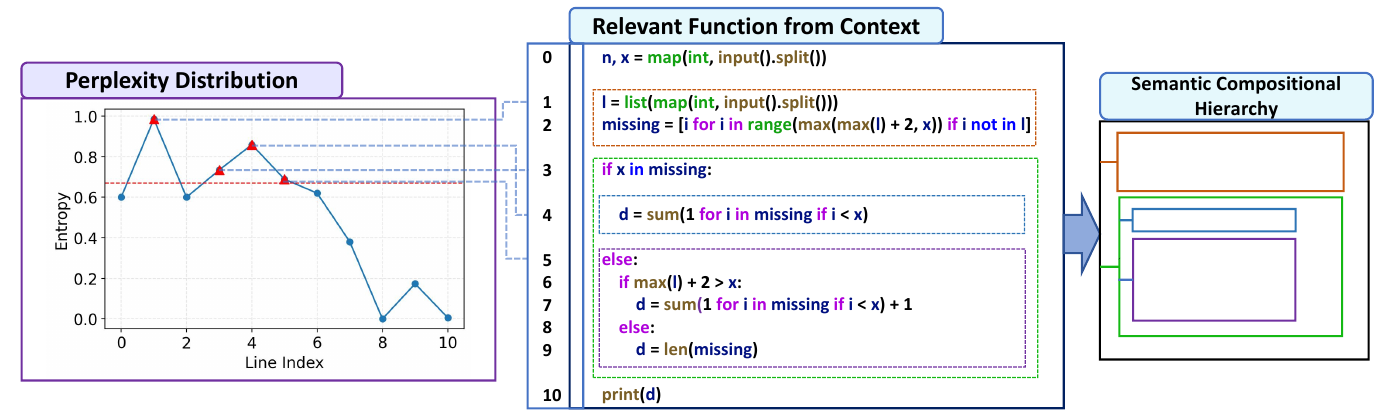}
    \caption{Hierarchical semantic decomposition example. Left: source code with token-entropy annotations, where color-coded regions indicate elevated LLM uncertainty. Right: the induced hierarchical semantic representation, with elements color-aligned to their semantic units of source code.}
    \label{fig:hierarchy_example}
\end{figure*}

To characterize code complexity from the perspective of large language models, we propose a structured semantic representation of source code. Programs are decomposed into entropy-aligned semantic units and recursively organized into a semantic compositional hierarchy, enabling the representation to capture how nonlinear program structures induce hierarchical dependencies during LLM inference.

Our decomposition strategy is guided by token-level entropy, which serves as a proxy for model uncertainty. Prior work by \citet{cooper2024perplexed} shows that high-entropy tokens correspond to positions where language models exhibit low confidence in next-token prediction, indicating increased processing difficulty. Complementary findings from \citet{shi2025longcodezip} further demonstrate that token-level perplexity effectively identifies semantically salient and structurally important regions in source code. These observations motivate the use of entropy signals to identify semantic boundaries from the perspective of the model.

The hierarchical semantic decomposition proceeds in four steps:

\textbf{(1) Code preprocessing.}
We first normalize the source code by removing comments and docstrings, producing a canonical representation that isolates executable program semantics.

\textbf{(2) Entropy computation.}
We compute token-level entropy using a pretrained language model. For each token $t_i$, the entropy is defined as
\begin{equation}
    H(t_i) = -\sum_j p(t_j \mid t_{<i}) \log p(t_j \mid t_{<i}),
\end{equation}
where the summation ranges over the entire vocabulary. Higher entropy values indicate greater predictive uncertainty at the corresponding token position.

\textbf{(3) Semantic unit decomposition.}
Semantic unit boundaries are identified using a hybrid entropy–structure criterion. A boundary is introduced when either (i) the token entropy exceeds a predefined threshold $\tau$, signaling a local increase in model uncertainty, or (ii) an explicit structural delimiter is encountered, such as the termination of a loop, conditional block, or function definition. 

\textbf{(4) Semantic compositional hierarchy construction.}
The resulting semantic units are recursively organized into a semantic compositional hierarchy through nested decomposition.
Units arising at the same compositional stage form branching structures that encode alternative semantic paths.
The final hierarchy therefore captures the nonlinear organization of program semantics in terms of both compositional levels and branching-induced semantic divergence.

Figure~\ref{fig:hierarchy_example} illustrates how localized token-level uncertainty is elevated into a structured semantic hierarchy.
Rather than treating uncertainty as a flat sequence-level signal, \ourmetric integrates entropy signals with structural delimiters to perform hierarchical semantic decomposition, thereby revealing the compositional and branching structures that underlie LLM code processing.

\subsection{Feature Extraction and Selection}
\label{sec:features}

Based on the semantic compositional hierarchy, we derive a set of quantitative features over semantic units to parameterize \ourmetric. 
These features are designed to characterize two fundamental dimensions of semantic nonlinearity that influence LLM reasoning: 
\emph{compositional level}, which captures the depth of semantic composition induced by nesting and scoping structures, and 
\emph{branching factor}, which captures semantic divergence introduced by branching structures.

Specifically, we define the following candidate features: 

\begin{itemize}[itemsep=0pt, topsep=0pt]
    \item \textit{Maximum Compositional Level (MaxCompLevel)}: the maximum compositional level within the hierarchy, representing the deepest nested semantic composition.
    \item \textit{Average Compositional Level (AvgCompLevel)}: the average compositional level across semantic units, reflecting typical hierarchical complexity.
    \item \textit{Total Compositional Level (TotalCompLevel)}: the cumulative sum of compositional levels over all semantic units, measuring the overall burden of hierarchical semantic composition.
    \item \textit{Maximum Branching Factor (MaxBranch)}: the largest branching factor observed in the hierarchy, characterizing the strongest horizontal expansion of semantic alternatives.
    \item \textit{Average Branching Factor (AvgBranch)}: the average branching factor across semantic units, quantifying the typical degree of local semantic divergence.
    \item \textit{Total Branching Factor (TotalBranch)}: the branching factor aggregated over the entire hierarchy, reflecting the global branching complexity.
        
\end{itemize}

To examine how these features relate to LLM-perceived difficulty, we perform partial correlation analysis following the protocol described in Section~\ref{sec:existing-metrics}, while controlling for code length. 
We compute group-level Spearman partial correlations between each feature and LLM task performance across multiple code-related tasks (Table~\ref{tab:main_results}).

\begin{table*}[t]
\centering
\caption{Partial Spearman correlations ($r$) of complexity features and \ourmetric with pass@1 across code tasks using \textsc{DeepSeek-V3}, while controlling for code length. ``--'' denotes non-significant correlations ($p \geq 0.05$).}
\label{tab:main_results}
\resizebox{\textwidth}{!}{
\begin{tabular}{llccccccc}
\toprule
\multirow{2}{*}{\textbf{Entropy Model}} & \multirow{2}{*}{\textbf{Code Task}}&  \multicolumn{3}{c}{\textbf{Compositional Level}} & \multicolumn{3}{c}{\textbf{Branching Factor}} & \multirow{2}{*}{\textbf{\ourmetric}}\\
        \cmidrule(lr){3-5} \cmidrule(lr){6-8}
& & \textbf{MaxCompLevel}  & \textbf{AvgCompLevel}  & \textbf{TotalCompLevel} & \textbf{MaxBranch} & \textbf{AvgBranch} & \textbf{TotalBranch} & \\

\midrule
\multirow{4}{*}{CodeLlama-7b}
& Program Repair      & - & - & -0.73 & -  & -  & -  & \textbf{-0.93} \\
& Code Translation    & - & - & -0.95 & -0.85 & - & -0.94 & \textbf{-0.97} \\
& Execution Reasoning & -0.72 & -0.79 & -0.84 & -0.69 & -0.75     & -0.85 & \textbf{-0.92} \\
& Code Summarization & - & - & -0.85 & -0.85 & - & -0.88 & \textbf{-0.90} \\

\midrule
\multirow{4}{*}{Deepseek-Coder-6.7b}
& Program Repair      & -0.80 & - & - & - & - & -0.79 & \textbf{-0.88} \\
& Code Translation    & - & - & -0.93 & -0.86 & - & -0.89 & \textbf{-0.94} \\
& Execution Reasoning & -0.75 & -0.81 & -0.87 & -0.72 & -0.79 & -0.86 & \textbf{-0.93} \\
& Code Summarization  & - & - & -0.86 & -0.87 & - & -0.90 & \textbf{-0.91} \\

\midrule
\multirow{4}{*}{Phi-3-mini-4k-instruct}
& Program Repair      & - & - & \textbf{-0.87} & -0.78 & - & - & -0.80 \\
& Code Translation    & - & - & \textbf{-0.95} & -0.88 & - & -0.83 & -0.88 \\
& Execution Reasoning & - & -0.77 & - & -0.89 & -0.81 & -0.90 & \textbf{-0.91} \\
& Code Summarization  & - & - & -0.92 & -0.90 & -0.69 & \textbf{-0.95} & -0.90 \\

\midrule
\multirow{4}{*}{Qwen2.5-Coder-1.5B}
& Program Repair      & - & -0.77 & -0.85 & - & - & - & \textbf{-0.87} \\
& Code Translation    & - & - & -0.85 & - & -0.86 & -0.82 & \textbf{-0.95} \\
& Execution Reasoning & - & -0.73 & \textbf{-0.95} & - & -0.78 & -0.92 & -0.84 \\
& Code Summarization  & - & -0.73 & -0.84 & -0.77 & - & -0.79 & \textbf{-0.85} \\

\midrule
\rowcolor{gray!12} \multicolumn{2}{c}{\textit{Significant Correlation Ratio}} & 3 / 16 & 6 / 16 & 14 / 16 & 11 / 16 & 6 / 16 & 13 / 16 & \textbf{16 / 16} \\

\bottomrule
\end{tabular}
}
\end{table*}

The results reveal clear task-dependent sensitivities to different forms of semantic nonlinearity. 
For program repair, \textit{TotalCompLevel} exhibits consistently strong negative partial correlations ($-0.73$ for CodeLlama-7b, $-0.87$ for Phi-3-mini-4k-instruct, and $-0.85$ for Qwen2.5-Coder-1.5B), indicating that increasingly deep semantic composition substantially degrades model performance. 
In contrast, code translation and execution reasoning are jointly influenced by compositional depth and branching complexity. 
For code translation, both \textit{TotalCompLevel} and \textit{TotalBranch} achieve moderate-to-strong correlations across all evaluated entropy models. 
Similarly, in execution reasoning, \textit{AvgCompLevel}, \textit{AvgBranch}, and \textit{TotalBranch} all demonstrate statistically significant correlations with model performance.

Despite these strong task-specific effects, no individual feature consistently exhibits robust predictive power across all task-model combinations, as each feature fails to achieve statistical significance under certain settings. 
To assess overall robustness, we compute the proportion of task-model combinations in which each feature achieves statistical significance. 
Among all candidates, the global measures \textit{TotalCompLevel} ($14/16$) and \textit{TotalBranch} ($13/16$) attain the highest coverage within the two core dimensions of semantic nonlinearity---\emph{compositional level} and \emph{branching factor}, respectively. 
These results suggest that cumulative semantic composition and global semantic divergence provide the most stable and representative characterization of LLM-perceived code complexity.


\subsection{Metric Definition}
\label{sec:metric}
Building on the complementary yet task-dependent behaviors of the extracted semantic features, we define \ourmetric as a unified measure of \emph{LLM-perceived code complexity} that jointly captures the two principal dimensions of semantic nonlinearity: hierarchical semantic composition and branching-induced semantic divergence.

Formally, given a semantic compositional hierarchy, \ourmetric is defined as a weighted combination of \emph{TotalBranch} and \emph{TotalCompLevel}. 
A weighting parameter $\alpha \in [0,1]$ controls the relative contribution of branching complexity and compositional complexity. 
The formulation naturally admits an equivalent per-unit decomposition, which makes explicit how local semantic effects accumulate across the hierarchy to produce global complexity.

\begin{equation}
\begin{aligned}
    \ourmetric 
    &= \alpha \cdot \text{TotalBranch} + (1-\alpha) \cdot \text{TotalCompLevel}, \\
    & =\sum_{v \in V} \left[ \alpha \cdot b(v) + (1-\alpha) \cdot d(v) \right],
\end{aligned}
\end{equation}
where $V$ denotes the set of semantic units in the compositional hierarchy, while $b(v)$ and $d(v)$ denote the branching factor and compositional level associated with unit $v$, respectively.

From an information-theoretic perspective, compositional level enlarges the contextual span over which branching-induced uncertainty must be propagated, while branching increases the number of competing semantic alternatives that the model must resolve.
\ourmetric captures the interaction between these two sources of complexity through a unified weighted formulation, enabling a stable and interpretable characterization of semantic nonlinearity across diverse programming tasks.


Notably, the semantic compositional hierarchy itself is constructed around regions of elevated token entropy, thereby grounding \ourmetric in the model’s intrinsic predictive uncertainty.
Unlike traditional complexity measures that are derived primarily from syntactic abstractions, \ourmetric directly incorporates model-sensitive semantic signals inferred from LLM behavior. 
Consequently, it provides a more faithful and cognitively aligned representation of code complexity from the perspective of large language models.

\subsection{Theoretical Justification}
\label{sec:theory-main}

The preceding empirical results show that \ourmetric correlates strongly with LLM performance even after controlling for code length, whereas traditional metrics such as cyclomatic complexity (CC) lose predictive power.  We now provide a theoretical account for this behavior.

Our central premise is that LLM predictive entropy accumulates systematically along two fundamental dimensions of program semantics: hierarchical semantic composition and branching-induced semantic divergence. 
These dimensions are explicitly modeled by \ourmetric, yet are only weakly characterized by traditional complexity measures such as CC.
Formally, CC measures the number of decision points through $\mathrm{CC}(c) = E - N + 2P$ (where $E$, $N$, and $P$ denote the numbers of edges, nodes, and connected components of the control-flow graph) without encoding the hierarchy and interaction of program semantics.
This distinction is formalized in Proposition~\ref{prop:cc-separation} (Appendix~\ref{sec:separation}), which shows that two programs with identical CC can exhibit $\Theta(n)$ versus $\Theta(n^2)$ growth in \ourmetric depending on whether their semantic structures are flat or deeply nested. This separation explains why CC loses predictive power under length control, whereas \ourmetric remains strongly associated with LLM difficulty.

To formalize this intuition, we analyze how predictive entropy accumulates over semantic hierarchies. 
Our analysis is based on two empirically motivated assumptions. 
Importantly, the significance of these assumptions lies not primarily in the specific penalty forms themselves, but in the theoretical consequences they enable: namely, the formal separation from CC (Proposition~\ref{prop:cc-separation}) and the monotonic relationship between \ourmetric and structural entropy (Corollary~\ref{cor:lmcc-proxy}).

\begin{assumption}[Depth-Dependent Context Degradation]
\label{assump:context}
Predictive uncertainty increases monotonically with compositional depth. For a semantic unit $v$ at depth $d(v)$, the entropy penalty satisfies $\Delta_{\text{depth}}(v) \geq \delta \cdot \max(0, d(v) - d^*)$, where $\delta > 0$ is the per-level degradation rate and $d^*$ is the critical depth beyond which context degradation occurs.
\end{assumption}

This assumption is supported by empirical findings in long-context LLM reasoning, including the ``lost in the middle'' phenomenon~\citep{liu2024lost}, attention sink effects~\citep{xiao2024efficient}, and broader intelligence degradation under extended contexts~\citep{wang2026intelligence}. 

\begin{assumption}[Branching-Induced Uncertainty]
\label{assump:branch}
At semantic divergence points with branching factor $b \geq 2$, the model incurs additional predictive uncertainty satisfying $\Delta H_{\text{branch}} \geq \gamma \cdot (b-1)$ for some $\gamma > 0$.
\end{assumption}

This assumption is motivated by entropy-aware reasoning~\citep{li2025entropy} and empirical studies consistently showing higher error rates in branching and nested code structures~\citep{sepidband2025enhancing,basu2025nestful,hu2025dynacode}. Intuitively, at a branch point the model's predictive distribution becomes a mixture over semantically divergent continuations, with each additional branch increasing the uncertainty that must be resolved.

\begin{theorem}[Hierarchical Entropy Accumulation]
\label{thm:entropy-accumulation}
Let $\mathcal{T} = (V, E)$ be a semantic hierarchy. Under Assumptions~\ref{assump:context} and~\ref{assump:branch}, the total predictive entropy satisfies
\begin{equation}
H_{\text{total}}(c) \geq \sum_{v \in V} H_0(v) + \Phi(\mathcal{T}),
\end{equation}
where $H_0(v)$ is the baseline entropy under ideal context access, and $\Phi(\mathcal{T}) = \delta \sum_{v} \max(0, d(v) - d^*) + \gamma \sum_{v} (b(v) - 1)^+$ is the structural penalty.
\end{theorem}

\begin{proof}
For each semantic unit $v \in V$, the conditional entropy decomposes as $H(v) = H_0(v) + \Delta_v$ where $\Delta_v \geq 0$. By Assumption~\ref{assump:context}, units at depth $d(v) > d^*$ incur penalty $\Delta_{\text{depth}}(v) \geq \delta(d(v) - d^*)$ due to context inaccessibility. By Assumption~\ref{assump:branch}, internal units with $b(v) \geq 2$ derived sub-units contribute $\Delta_{\text{branch}}(v) \geq \gamma(b(v) - 1)$. Since these penalties address orthogonal aspects (context span vs.\ path resolution), aggregating yields $H_{\text{total}} = \sum_v H(v) \geq \sum_v H_0(v) + \Phi(\mathcal{T})$.
\end{proof}

\begin{corollary}[\ourmetric as Structural Complexity Proxy]
\label{cor:lmcc-proxy}
The structural penalty satisfies $\Phi(\mathcal{T}) \leq \delta \cdot \text{TotalCompLevel} + \gamma \cdot \text{TotalBranch}$, with equality (up to constants) when all units exceed the critical depth. Thus, \ourmetric $= \alpha \cdot \text{TotalBranch} + (1-\alpha) \cdot \text{TotalCompLevel}$ is monotonically related to the structural entropy penalty $\Phi(\mathcal{T})$.
\end{corollary}

\begin{proof}
From $\max(0, d(v) - d^*) \leq d(v)$ and $(b(v)-1)^+ \leq b(v)$, the upper bound follows. When $d(v) > d^*$ for all $v$, we have $\Phi = \delta(\text{TotalCompLevel} - d^*|V|) + \gamma(\text{TotalBranch} - |V|)$, establishing the monotonic relationship.
\end{proof}

We note that the additive form provides a conservative lower bound on the true structural penalty; in particular, deeper branching points naturally receive larger contributions through elevated $d(v)$, implicitly reflecting depth--branching interaction.
Additional theoretical analyses, including formal separation from cyclomatic complexity and proofs of orthogonality to code length, are provided in Appendix~\ref{sec:theory}.

\section{Experiments}
\label{sec:experiments}

We evaluate \ourmetric through both correlational analysis and causal intervention studies.

\subsection{Correlation with Model Performance}
\label{sec:effectiveness}

We first investigate whether \ourmetric effectively characterizes code complexity from the perspective of large language models. Following the datasets and experimental protocol adopted in our empirical evaluation of existing complexity metrics, we compute partial correlation coefficients between \ourmetric and LLM task performance while controlling for code length. 
Token entropy is estimated using four language models spanning different architectures and parameter scales: \textsc{CodeLlama-7b-hf}\footnote{https://huggingface.co/meta-llama/CodeLlama-7b-hf}, \textsc{Phi-3-mini-4k-instruct}\footnote{https://huggingface.co/microsoft/Phi-3-mini-4k-instruct}, \textsc{DeepSeek-Coder-6.7B-base}\footnote{https://huggingface.co/deepseek-ai/deepseek-coder-6.7b-base}, and \textsc{Qwen2.5-Coder-1.5B}\footnote{https://huggingface.co/Qwen/Qwen2.5-Coder-1.5B}. 
The semantic segmentation threshold is fixed at $\tau = 0.67$ following prior work~\citep{wang2025beyond}, while the weighting factor is set to $\alpha = 0.8$ based on the ablation results presented in Appendix~\ref{sec:ablation}. 
Additional results obtained with other task-performing LLMs are reported in Appendix~\ref{sec:main-results-sup}.

\textbf{Overall Effectiveness.}
As shown in the final column of Table~\ref{tab:main_results}, \ourmetric exhibits strong and consistent negative correlations with LLM performance across all evaluated tasks and entropy estimation models. Correlation coefficients range from $-0.80$ to $-0.97$, indicating that higher \ourmetric values are reliably associated with lower task performance.

\textbf{Statistical Robustness Across Settings.}
Although \ourmetric does not achieve absolute numerical superiority under every individual setting, results derived from all four entropy estimation models remain consistently statistically significant across all evaluated tasks. Specifically, \ourmetric achieves statistically significant correlations in all 16 experimental settings, outperforming its two constituent dimensions, \textit{TotalCompLevel} ($14/16$) and \textit{TotalBranch} ($13/16$). These results demonstrate the strong robustness and generalizability of \ourmetric across diverse experimental conditions, which is a critical requirement for a reliable evaluation metric.

\textbf{Cross-Task Generalization.} 
Across diverse code intelligence tasks, \ourmetric consistently demonstrates strong negative correlations with model performance, suggesting that it effectively captures the intrinsic difficulty of code as perceived by large language models. Furthermore, parameter settings calibrated on three tasks, namely program repair, code translation, and execution reasoning, transfer effectively to the code summarization task without additional tuning, further validating the strong cross-task generalization ability of \ourmetric.

\textbf{Effect of Entropy Models.} 
The effectiveness of \ourmetric varies moderately across entropy estimation models. Under evaluation with \textsc{DeepSeek-V3}, \textsc{CodeLlama-7b-hf} and \textsc{DeepSeek-Coder-6.7b-base} consistently achieve the strongest overall performance across all tasks, both in terms of correlation strength and statistical effectiveness. In contrast, \textsc{Phi-3-mini-4k-instruct} achieves the best performance only on the execution reasoning benchmark, while \textsc{Qwen2.5-Coder-1.5B} does not attain the highest absolute correlation coefficient on that task.
We hypothesize that larger base models provide better-calibrated token probability estimates across diverse code patterns, leading to more robust complexity correlations.
By comparison, smaller models may lack sufficient capacity to capture fine-grained complexity variations.

Overall, \ourmetric demonstrates stable and statistically reliable correlations with LLM performance, substantially outperforming traditional complexity measures. 
By jointly integrating compositional level and branching factor within an entropy-guided semantic representation, \ourmetric provides a principled characterization of code complexity from the perspective of large language models.

\subsection{Impact on LLM-Based Code Tasks}
\label{sec:improvement}

\begin{table*}[t]
\centering
\caption{Impact of semantics-preserving complexity reduction on LLM-based code tasks.
Performance (pass@1) is evaluated on the same subset of samples before and after rewriting.}
\label{tab:impact-results}
\begin{tabular}{lcccccc}
\toprule
\multirow{2}{*}{Task} 
& \multicolumn{2}{c}{Performance ($\uparrow$)} 
& \multicolumn{2}{c}{\ourmetric ($\downarrow$)} 
& \multicolumn{2}{c}{Cyclomatic ($\downarrow$)} \\
\cmidrule(lr){2-3} \cmidrule(lr){4-5} \cmidrule(lr){6-7}
& Original & Rewritten 
& Original & Rewritten 
& Original & Rewritten \\
\midrule
Program Repair        & 13.4\% & \textbf{16.2\%}$^{\dagger}$ &  77.7 & \textbf{72.5} & 18.3 & 18.3 \\
Code Translation      & 42.2\% & \textbf{46.5\%}$^{\dagger}$ &  45.9 &  \textbf{40.9} & 11.0 & 11.0 \\
Execution Reasoning   & 94.1\% & \textbf{95.0\%}$^{\dagger}$ &   12.9 & \textbf{9.3} & 4.1 & 4.1 \\
\bottomrule
\end{tabular}

\end{table*}

Beyond correlation analysis, we further assess the practical utility of \ourmetric through a controlled causal intervention: 
\emph{Can reducing \ourmetric directly improve LLM performance on code-related tasks?}
To answer this question, we perform semantics-preserving code rewriting to reduce \ourmetric and examine the resulting changes in downstream task performance.

Specifically, we employ \textsc{DeepSeek-V3.2}~\citep{deepseekai2024deepseekv3} to rewrite code samples under two constraints: 
(1) the rewritten program must preserve functional equivalence by passing all original test cases (except for the program repair task), and 
(2) the rewritten version must achieve a strictly lower \ourmetric than the original code.
To further disentangle the effect of \ourmetric from conventional structural simplifications, we additionally constrain cyclomatic complexity to remain unchanged or increase during rewriting. This ensures that any observed performance gains cannot be attributed to reductions in traditional control-flow complexity measures.
Additional details of the rewriting pipeline are provided in Appendix~\ref{sec:code-rewrite}.

The resulting rewritten dataset therefore contains semantically equivalent programs with systematically reduced \ourmetric values while maintaining comparable traditional structural complexity. We subsequently reevaluate all downstream tasks on the subset of samples admitting valid rewrites, using identical evaluation models and experimental settings before and after rewriting. This paired and controlled setup enables direct analysis of the causal impact of reducing \ourmetric on LLM performance.

As shown in Table~\ref{tab:impact-results}, reducing \ourmetric consistently yields statistically significant performance improvements across all evaluated tasks. Specifically, program repair, code translation, and execution reasoning achieve relative improvements of $20.9\%$ (from $13.4\%$ to $16.2\%$), $10.2\%$ (from $42.2\%$ to $46.5\%$), and $0.96\%$ (from $94.1\%$ to $95.0\%$), respectively.
Notably, the program repair benchmark exhibits the largest improvement after rewriting. This task also has the highest original \ourmetric values and the lowest baseline performance, suggesting that semantically complex code imposes substantially greater reasoning difficulty on language models and therefore offers larger room for improvement when complexity is reduced. In contrast, the execution reasoning benchmark consists primarily of relatively short programs with already near-saturated baseline accuracy, naturally limiting the achievable gains from complexity reduction.


\section{Related Work}

\subsection{LLMs for Code}

Recent years have witnessed rapid progress in LLMs for code. 
Codex~\citep{chen2021evaluating} first demonstrated that scaling language models on large code corpora enables strong performance in code generation and comprehension. 
Subsequently, open-source code LLMs such as Code Llama~\citep{roziere2023codellama}, StarCoder~\citep{lozhkov2024starcoder2}, DeepSeek-Coder~\citep{zhu2024deepseekcoderv2}, and Qwen2.5-Coder~\citep{hui2024qwen25coder} achieved competitive performance across diverse coding benchmarks.

Beyond generation, LLMs have been extensively studied on code-centric tasks including program repair~\citep{muennighoff2023octopack,shi2024code}, code translation~\citep{khan2024xcodeeval,wang2025evoc2rust}, and execution reasoning~\citep{gu2024cruxeval}. 
To support systematic evaluation, numerous benchmarks have been introduced, such as HumanEval~\citep{chen2021evaluating}, xCodeEval~\citep{khan2024xcodeeval}, LiveCodeBench~\citep{jain2024livecodebench}, SWE-bench~\citep{jimenez2024swebench}, and BigCodeBench~\citep{zhuo2025bigcodebench}.

While these works demonstrate strong empirical capabilities of LLMs for code, they primarily focus on performance gains and benchmark comparisons, with limited attention to modeling the intrinsic difficulty of code from the perspective of LLMs. This motivates further investigation into the structural factors underlying LLM performance on code tasks.

\subsection{Code Complexity Metrics}

A broad spectrum of code complexity metrics has been proposed to quantify program difficulty~\citep{HaoHMDLCM26}. Classic examples include \emph{cyclomatic complexity}~\citep{mccabe1976complexity}, which measures the number of linearly independent execution paths via control-flow analysis, and \emph{Halstead’s metrics}~\citep{1977Elements}, which characterize complexity through the distribution of operators and operands. To better approximate human code comprehension, \emph{Cognitive Complexity}~\citep{campbell2018cognitive} penalizes nested control structures and control-flow interruptions. Beyond static analyses, recent work has also leveraged neurophysiological and biometric signals, such as EEG-based measures of cognitive load~\citep{hao2023accuracy} and fMRI-based analyses of brain activation~\citep{9402005}, to assess human mental effort during programming tasks. 

More recently, several studies have explored the relationship between program complexity and LLM behavior. Sepidband et al.~\citep{sepidband2025enhancing} reported correlations between existing complexity metrics and LLM code generation performance and proposed complexity-aware feedback mechanisms. Other work investigated how program structure and complexity affect LLM execution reasoning~\citep{liu2025tool}. Related systems such as EpiCoder~\citep{WangL000HGHX0S025} explicitly control generated code complexity via feature trees, while evaluation frameworks including CodeMind~\citep{liu2026codemind} and DynaCode~\citep{hu2025dynacode} consistently observe performance degradation as code complexity increases.

Despite these advances, existing studies still rely primarily on traditional complexity metrics and implicitly assume that such measures adequately reflect the difficulty encountered by LLMs. No prior work has introduced a complexity metric explicitly designed to model how LLMs perceive and process code. Our work addresses this gap by proposing a model-aware complexity metric grounded in the representational and computational characteristics of large language models.

\section{Conclusion}

In this paper, we show that traditional code complexity metrics do not reliably reflect the difficulty that large language models encounter on code-related tasks. To address this issue, we propose \ourmetric, a model-aware code complexity metric that uses token-level entropy as a proxy for model uncertainty to quantify semantic nonlinearity induced by the compositional and branching structure of semantic units. 
We evaluate LM-CC through extensive experiments across diverse code understanding and generation tasks. The results demonstrate that \ourmetric correlates substantially more strongly with LLM performance than existing metrics. Furthermore, semantics-preserving transformations that reduce \ourmetric consistently improve downstream task performance, highlighting its effectiveness and practical relevance.

\section*{Acknowledgements}
This research is funded by the National Key Research and Development Program of China (Grant No. 2023YFB4503802), the National Natural Science Foundation of China (Grant No. 62232003), and the Natural Science Foundation of Shanghai (Grant No. 25ZR1401175).

\section*{Impact Statement}

In this work, we reveal that traditional code complexity metrics are grounded in human cognitive load and exhibit no consistent correlation with LLM performance. In response, we propose LM-CC, a LLM-centric code complexity metric designed to quantify the difficulty LLMs experience when processing code.
It decomposes programs into entropy-aligned semantic units, arranges them into a compositional hierarchy, and quantifies complexity by aggregating compositional depth and branching-induced divergence.
LM-CC achieves significantly stronger and more consistent correlations with task performance than traditional metrics, providing a unified model-aware code complexity standard.

\nocite{langley00}

\bibliography{ref}

@inproceedings{sepidband2025enhancing,
  author       = {Melika Sepidband and
                  Hamed Taherkhani and
                  Song Wang and
                  Hadi Hemmati},
  title        = {Enhancing {LLM}-Based Code Generation with Complexity Metrics: {A} Feedback-Driven
                  Approach},
  booktitle    = {49th {IEEE} Annual Computers, Software, and Applications Conference,
                  {COMPSAC} 2025, Toronto, ON, Canada, July 8-11, 2025},
  pages        = {1416--1426},
  publisher    = {{IEEE}},
  year         = {2025},
}

@article{mccabe1976complexity,
  title={A complexity measure},
  author={McCabe, Thomas J},
  journal={IEEE Transactions on software Engineering},
  number={4},
  pages={308--320},
  year={1976},
  publisher={IEEE}
}

@techreport{campbell2018cognitive,
  author       = {Campbell, G. Ann},
  title        = {Cognitive Complexity: A New Way of Measuring Understandability},
  institution  = {SonarSource SA},
  year         = {2018},
  type         = {White Paper},
  url          = {https://www.sonarsource.com/docs/CognitiveComplexity.pdf}
}

@article{chen2021evaluating,
  title     = {Evaluating Large Language Models Trained on Code},
  author    = {Chen, Mark and Tworek, Jerry and Jun, Heewoo and Yuan, Qiming and Pinto, Henrique Ponde de Oliveira and Kaplan, Jared and Edwards, Harri and Burda, Yuri and Joseph, Nicholas and Brockman, Greg and others},
  journal   = {arXiv preprint arXiv:2107.03374},
  year      = {2021}
}

@article{muennighoff2023octopack,
  title     = {{OctoPack}: Instruction Tuning Code Large Language Models},
  author    = {Muennighoff, Niklas and Liu, Qian and Zebaze, Armel and Zheng, Qinkai and Hui, Binyuan and Zhuo, Terry Yue and Singh, Swayam and Tang, Xiangru and von Werra, Leandro and Longpre, Shayne},
  journal   = {arXiv preprint arXiv:2308.07124},
  year      = {2023}
}

@inproceedings{khan2024xcodeeval,
  title     = {{XCodeEval}: An Execution-based Large Scale Multilingual Multitask Benchmark for Code Understanding, Generation, Translation and Retrieval},
  author    = {Khan, Mohammad Abdullah Matin and Bari, M. Saiful and Do, Xuan Long and Wang, Weishi and Parvez, Md Rizwan and Joty, Shafiq},
  booktitle = {Proceedings of the 62nd Annual Meeting of the Association for Computational Linguistics (Volume 1: Long Papers)},
  pages     = {6766--6805},
  year      = {2024},
  publisher = {Association for Computational Linguistics}
}

@inproceedings{gu2024cruxeval,
  title     = {{CRUXEval}: A Benchmark for Code Reasoning, Understanding and Execution},
  author    = {Gu, Alex and Rozi{\`e}re, Baptiste and Leather, Hugh and Solar-Lezama, Armando and Synnaeve, Gabriel and Wang, Sida},
  booktitle = {Proceedings of the 41st International Conference on Machine Learning},
  pages     = {16568--16621},
  year      = {2024},
  volume    = {235},
  series    = {Proceedings of Machine Learning Research},
  publisher = {PMLR}
}

@inproceedings{liu2025tool,
  author    = {Liu, Changshu and Jabbarvand, Reyhaneh},
  title     = {A Tool for In-depth Analysis of Code Execution Reasoning of Large Language Models},
  booktitle = {Proceedings of the 33rd ACM International Conference on the Foundations of Software Engineering (FSE Companion)},
  pages     = {1178--1182},
  year      = {2025},
  publisher = {ACM}
}

@inproceedings{hu2025dynacode,
  title     = {{DynaCode}: A Dynamic Complexity-Aware Code Benchmark for Evaluating Large Language Models in Code Generation},
  author    = {Hu, Wenhao and Duan, Jinhao and Wei, Chunchen and Zhang, Li and Zhang, Yue and Xu, Kaidi},
  booktitle = {Findings of the Association for Computational Linguistics: ACL 2025},
  pages     = {21980--21997},
  year      = {2025},
  publisher = {Association for Computational Linguistics}
}

@article{liu2026codemind,
  title={Code{M}ind: Evaluating Large Language Models for Implicit and Explicit Code Execution Reasoning},
  author={Liu, Changshu and Chen, Yang and Jabbarvand, Reyhaneh},
  journal={IEEE Transactions on Software Engineering},
  year={2026},
  publisher={IEEE}
}

@inproceedings{zhuo2025bigcodebench,
  title     = {{BigCodeBench}: Benchmarking Code Generation with Diverse Function Calls and Complex Instructions},
  author    = {Zhuo, Terry Yue and Vu, Minh Chien and Chim, Jenny and Hu, Han and Yu, Wenhao and Widyasari, Ratnadira and Yusuf, Imam Nur Bani and Zhan, Haolan and He, Junda and Paul, Indraneil and Brunner, Simon and Gong, Chen and Hoang, Thong and Zebaze, Armel Randy and Hong, Xiaoheng and Li, Wen-Ding and Kaddour, Jean and Xu, Ming and Zhang, Zhihan and Yadav, Prateek and Jain, Naman and Gu, Alex and Cheng, Zhoujun and Liu, Jiawei and Liu, Qian and Wang, Zijian and Lo, David and Hui, Binyuan and Muennighoff, Niklas and Fried, Daniel and Du, Xiaoning and de Vries, Harm and von Werra, Leandro},
  booktitle = {Proceedings of the 13th International Conference on Learning Representations},
  year      = {2025}
}

@article{jain2024livecodebench,
  title     = {{LiveCodeBench}: Holistic and Contamination Free Evaluation of Large Language Models for Code},
  author    = {Jain, Naman and Han, King and Gu, Alex and Li, Wen-Ding and Yan, Fanjia and Zhang, Tianjun and Wang, Sida and Solar-Lezama, Armando and Sen, Koushik and Stoica, Ion},
  journal   = {arXiv preprint arXiv:2403.07974},
  year      = {2024}
}

@article{zhu2024deepseekcoderv2,
  title     = {{DeepSeek-Coder-V2}: Breaking the Barrier of Closed-Source Models in Code Intelligence},
  author    = {Zhu, Qihao and Guo, Daya and Shao, Zhihong and Yang, Dejian and Wang, Peiyi and Xu, Runxin and Wu, Y. and Li, Yukun and Gao, Huazuo and Ma, Shirong and others},
  journal   = {arXiv preprint arXiv:2406.11931},
  year      = {2024}
}

@article{hui2024qwen25coder,
  title     = {{Qwen2.5-Coder} Technical Report},
  author    = {Hui, Binyuan and Yang, Jian and Cui, Zeyu and Yang, Jiaxi and Liu, Dayiheng and Zhang, Lei and Liu, Tianyu and Zhang, Jiajun and Yu, Bowen and Lu, Keming and others},
  journal   = {arXiv preprint arXiv:2409.12186},
  year      = {2024}
}

@article{roziere2023codellama,
  title     = {{Code Llama}: Open Foundation Models for Code},
  author    = {Rozi{\`e}re, Baptiste and Gehring, Jonas and Gloeckle, Fabian and Sootla, Sten and Gat, Itai and Tan, Xiaoqing Ellen and Adi, Yossi and Liu, Jingyu and Remez, Tal and Rapin, J{\'e}r{\'e}my and others},
  journal   = {arXiv preprint arXiv:2308.12950},
  year      = {2023}
}

@article{lozhkov2024starcoder2,
  title     = {{StarCoder 2} and {The Stack v2}: The Next Generation},
  author    = {Lozhkov, Anton and Li, Raymond and Allal, Loubna Ben and Cassano, Federico and Lamy-Poirier, Joel and Tazi, Nouamane and Tang, Ao and Pykhtar, Dmytro and Liu, Jiawei and Wei, Yuxiang and others},
  journal   = {arXiv preprint arXiv:2402.19173},
  year      = {2024}
}

@inproceedings{jimenez2024swebench,
  title     = {{SWE-bench}: Can Language Models Resolve Real-World {GitHub} Issues?},
  author    = {Jimenez, Carlos E. and Yang, John and Wettig, Alexander and Yao, Shunyu and Pei, Kexin and Press, Ofir and Narasimhan, Karthik},
  booktitle = {Proceedings of the 12th International Conference on Learning Representations},
  year      = {2024}
}

@inproceedings{shi2025longcodezip,
  author       = {Yuling Shi and
                  Yichun Qian and
                  Hongyu Zhang and
                  Beijun Shen and
                  Xiaodong Gu},
  title        = {{LongCodeZip}: Compress Long Context for Code Language Models},
  booktitle    = {40th {IEEE/ACM} International Conference on Automated Software Engineering,
                  {ASE} 2025, November 16-20, 2025},
  pages        = {141--153},
  publisher    = {{IEEE}},
  year         = {2025},
}

@inproceedings{shi2024code,
  title     = {From Code to Correctness: Closing the Last Mile of Code Generation with Hierarchical Debugging},
  author    = {Shi, Yuling and Wang, Songsong and Wan, Chengcheng and Wang, Min and Gu, Xiaodong},
  booktitle    = {2026 IEEE/ACM 47th International Conference on Software Engineering (ICSE)},
  year      = {2026}
}

@inproceedings{shi2024between,
  title        = {Between Lines of Code: Unraveling the Distinct Patterns of Machine and Human Programmers},
  author       = {Shi, Yuling and Zhang, Hongyu and Wan, Chengcheng and Gu, Xiaodong},
  booktitle    = {2025 IEEE/ACM 47th International Conference on Software Engineering (ICSE)},
  pages        = {51--62},
  year         = {2025},
  organization = {IEEE Computer Society}
}

@inproceedings{wang2025evoc2rust,
  title={{EvoC2Rust}: A Skeleton-guided Framework for Project-Level C-to-Rust Translation},
  author={Wang, Chaofan and Yu, Tingrui and Xie, Chen and Wang, Jie and Chen, Dong and Zhang, Wenrui and Shi, Yuling and Gu, Xiaodong and Shen, Beijun},
  booktitle    = {2026 IEEE/ACM 47th International Conference on Software Engineering (ICSE) SEIP},
  year={2026}
}

@inproceedings{anand2024critical,
  title     = {A Critical Study of What Code-{LLMs} (Do Not) Learn},
  author    = {Anand, Abhinav and Verma, Shweta and Narasimhan, Krishna and Mezini, Mira},
  booktitle = {Findings of the Association for Computational Linguistics: ACL 2024},
  pages     = {15869--15889},
  year      = {2024},
  publisher = {Association for Computational Linguistics}
}

@article{cooper2024perplexed,
  title     = {Perplexed: Understanding When Large Language Models are Confused},
  author    = {Cooper, Nathan A. and Scholak, Torsten},
  journal   = {arXiv preprint arXiv:2404.06634},
  year      = {2024}
}

@article{jiang2025survey,
  title     = {A Survey on Large Language Models for Code Generation},
  author    = {Jiang, Juyong and Wang, Fan and Shen, Jiasi and Kim, Sungju and Kim, Sunghun},
  journal   = {ACM Transactions on Software Engineering and Methodology},
  year      = {2025},
  publisher = {ACM}
}

@inproceedings{du2025context,
  title     = {Context Length Alone Hurts {LLM} Performance Despite Perfect Retrieval},
  author    = {Du, Yufeng and Tian, Minyang and Ronanki, Srikanth and Rongali, Subendhu and Bodapati, Sravan Babu and Galstyan, Aram and Wells, Azton and Schwartz, Roy and Huerta, Eliu A. and Peng, Hao},
  booktitle = {Findings of the Association for Computational Linguistics: EMNLP 2025},
  pages     = {23281--23298},
  year      = {2025},
  publisher = {Association for Computational Linguistics}
}

@article{deepseekai2024deepseekv3,
  title     = {{DeepSeek-V3} Technical Report},
  author    = {{DeepSeek-AI}},
  journal   = {arXiv preprint arXiv:2412.19437},
  year      = {2024}
}

@article{Weyuker88,
  author       = {Elaine J. Weyuker},
  title        = {Evaluating Software Complexity Measures},
  journal      = {{IEEE} Trans. Software Eng.},
  volume       = {14},
  number       = {9},
  pages        = {1357--1365},
  year         = {1988},
}

@article{AjamiWF19,
  author       = {Shulamyt Ajami and
                  Yonatan Woodbridge and
                  Dror G. Feitelson},
  title        = {Syntax, predicates, idioms - what really affects code complexity?},
  journal      = {Empir. Softw. Eng.},
  volume       = {24},
  number       = {1},
  pages        = {287--328},
  year         = {2019},
}

@article{Feitelson23,
  author       = {Dror G. Feitelson},
  title        = {From Code Complexity Metrics to Program Comprehension},
  journal      = {Commun. {ACM}},
  volume       = {66},
  number       = {5},
  pages        = {52--61},
  year         = {2023},
}

@inproceedings{GirjoabaC24,
  author       = {Andrei V. Girjoaba and
                  Andrea Capiluppi},
  title        = {Refactoring Legacy Code Using Cleaning Up Cycles: An Experience Report},
  booktitle    = {{IEEE} International Conference on Software Maintenance and Evolution,
                  {ICSME} 2024, Flagstaff, AZ, USA, October 6-11, 2024},
  pages        = {753--764},
  publisher    = {{IEEE}},
  year         = {2024},
}

@article{BenarochL23,
  author       = {Michel Benaroch and
                  Kalle Lyytinen},
  title        = {How Much Does Software Complexity Matter for Maintenance Productivity?
                  The Link Between Team Instability and Diversity},
  journal      = {{IEEE} Trans. Software Eng.},
  volume       = {49},
  number       = {4},
  pages        = {2459--2475},
  year         = {2023},
}

@inproceedings{AlqadiM20,
  author       = {Basma S. Alqadi and
                  Jonathan I. Maletic},
  title        = {Slice-Based Cognitive Complexity Metrics for Defect Prediction},
  booktitle    = {27th {IEEE} International Conference on Software Analysis, Evolution
                  and Reengineering, {SANER} 2020, London, ON, Canada, February 18-21,
                  2020},
  pages        = {411--422},
  publisher    = {{IEEE}},
  year         = {2020},

}

@inproceedings{WangL000HGHX0S025,
  author       = {Yaoxiang Wang and
                  Haoling Li and
                  Xin Zhang and
                  Jie Wu and
                  Xiao Liu and
                  Wenxiang Hu and
                  Zhongxin Guo and
                  Yangyu Huang and
                  Ying Xin and
                  Yujiu Yang and
                  Jinsong Su and
                  Qi Chen and
                  Scarlett Li},
  title        = {{EpiCoder}: Encompassing Diversity and Complexity in Code Generation},
  booktitle    = {Forty-second International Conference on Machine Learning, {ICML}
                  2025, Vancouver, BC, Canada, July 13-19, 2025},
  publisher    = {OpenReview.net},
  year         = {2025},
}

@article{HaoHMDLCM26,
  author       = {Hao Gao and
                  Haytham Hijazi and
                  J{\'{u}}lio Medeiros and
                  Jo{\~{a}}o Dur{\~{a}}es and
                  Chan{-}Tong Lam and
                  Paulo de Carvalho and
                  Henrique Madeira},
  title        = {Complementarity in software code complexity metrics},
  journal      = {J. Syst. Softw.},
  volume       = {232},
  pages        = {112679},
  year         = {2026},
}

@book{1977Elements,
  title     = {Elements of Software Science},
  author    = {Halstead, M. H.},
  year      = {1977},
  publisher = {Elsevier Science Inc.},
  series    = {Operating and Programming Systems Series}
}

@INPROCEEDINGS{9402005,
  author={Peitek, Norman and Apel, Sven and Parnin, Chris and Brechmann, André and Siegmund, Janet},
  booktitle={2021 IEEE/ACM 43rd International Conference on Software Engineering (ICSE)}, 
  title={Program Comprehension and Code Complexity Metrics: An f{MRI} Study}, 
  year={2021},
  pages={524-536},
}

@article{hao2023accuracy,
  title   = {On the Accuracy of Code Complexity Metrics: A Neuroscience-Based Guideline for Improvement},
  author  = {Hao, Guozhu and Hijazi, Hadi and Dur{\~a}es, Jo{\~a}o and Medeiros, Jo{\~a}o and Couceiro, Rui and Lam, Chi T. and Teixeira, Carlos and Castelhano, Jorge and Castelo-Branco, Miguel and Carvalho, Paulo and others},
  journal = {Frontiers in Neuroscience},
  volume  = {16},
  pages   = {1065366},
  year    = {2023}
}

@article{OMAN1994251,
title = {Construction and testing of polynomials predicting software maintainability},
journal = {Journal of Systems and Software},
volume = {24},
number = {3},
pages = {251-266},
year = {1994},
author = {Paul Oman and Jack Hagemeister}
}

@article{wang2025beyond,
  title={Beyond the 80/20 rule: High-entropy minority tokens drive effective reinforcement learning for {LLM} reasoning},
  author={Wang, Shenzhi and Yu, Le and Gao, Chang and Zheng, Chujie and Liu, Shixuan and Lu, Rui and Dang, Kai and Chen, Xionghui and Yang, Jianxin and Zhang, Zhenru and others},
  journal={arXiv preprint arXiv:2506.01939},
  year={2025}
}

@article{liu2024lost,
  title={Lost in the Middle: How Language Models Use Long Contexts},
  author={Liu, Nelson F. and Lin, Kevin and Hewitt, John and Paranjape, Ashwin and Bevilacqua, Michele and Petroni, Fabio and Liang, Percy},
  journal={Transactions of the Association for Computational Linguistics},
  volume={12},
  pages={157--173},
  year={2024},
  publisher={MIT Press}
}

@inproceedings{xiao2024efficient,
  title={Efficient Streaming Language Models with Attention Sinks},
  author={Xiao, Guangxuan and Tian, Yuandong and Chen, Beidi and Han, Song and Lewis, Mike},
  booktitle={Proceedings of the 12th International Conference on Learning Representations},
  year={2024}
}

@article{li2025entropy,
  title={Entropy-gated branching for efficient test-time reasoning},
  author={Li, Xianzhi and Callanan, Ethan and Ghassel, Abdellah and Zhu, Xiaodan},
  journal={arXiv preprint arXiv:2503.21961},
  year={2025}
}

@article{wang2026intelligence,
  title={Intelligence Degradation in Long-Context LLMs: Critical Threshold Determination via Natural Length Distribution Analysis},
  author={Wang, Weiwei and Min, Jiyong and Zou, Weijie},
  journal={arXiv preprint arXiv:2601.15300},
  year={2026}
}

@inproceedings{basu2025nestful,
  title={Nestful: A benchmark for evaluating llms on nested sequences of api calls},
  author={Basu, Kinjal and Abdelaziz, Ibrahim and Kate, Kiran and Agarwal, Mayank and Crouse, Maxwell and Rizk, Yara and Bradford, Kelsey and Munawar, Asim and Kumaravel, Sadhana and Goyal, Saurabh and others},
  booktitle={Proceedings of the 2025 Conference on Empirical Methods in Natural Language Processing},
  pages={33526--33535},
  year={2025}
}

@inproceedings{LeVPLNB25,
  author       = {Cuong Chi Le and
                  Hoang Chau Truong Vinh and
                  Huy Nhat Phan and
                  Dung Duy Le and
                  Tien N. Nguyen and
                  Nghi D. Q. Bui},
  title        = {{VisualCoder}: Guiding Large Language Models in Code Execution with
                  Fine-grained Multimodal Chain-of-Thought Reasoning},
  booktitle    = {Findings of the Association for Computational Linguistics: {NAACL}
                  2025, Albuquerque, New Mexico, USA, April 29 - May 4, 2025},
  series       = {Findings of {ACL}},
  volume       = {{NAACL} 2025},
  pages        = {6628--6645},
  publisher    = {Association for Computational Linguistics},
  year         = {2025},

}
\bibliographystyle{icml2026}

\newpage
\appendix
\onecolumn
{\LARGE \textbf{Appendix}}
\section{Pseudocode of \ourmetric}
\label{sec:algorithm}

Algorithm~\ref{algo:metric} presents the computation procedure of \ourmetric.
Given a source code snippet, \ourmetric first preprocesses the input and computes token-level entropy using a pretrained language model (Lines~1--4).
Semantic unit boundaries are then identified using two complementary signals: positions where model uncertainty exceeds the threshold $\tau$, and syntactic delimiters such as loops or function boundaries (Lines~5--9). This hybrid criterion allows the segmentation to capture both model-perceived difficulty and program structure.
The algorithm subsequently constructs a semantic compositional hierarchy through a BFS traversal (Lines~10--15). Starting from a top-level unit, each unit range is recursively partitioned according to detected boundaries and indentation structure, and the resulting sub-units are enqueued for further processing.
After constructing the hierarchy, \ourmetric is computed by aggregating contributions from all semantic units (Lines~16--19), where each unit is weighted according to its branching factor and compositional level. Consequently, semantically deeper and more branching structures receive higher complexity scores.

\begin{algorithm}[ht]
\caption{Computation of \ourmetric}
\label{algo:metric}
\small
\KwIn{Source code $P$, entropy threshold $\tau$, smoothing constant $\alpha$}
\KwOut{Complexity score $\ourmetric(P)$}

$P' \gets \textsc{RemoveComments}(P)$\;
$(t_1, \ldots, t_n) \gets \textsc{Tokenize}(P')$\;

\ForEach{token $t_i$}{
    $H(t_i) \gets -\sum_{j} p(t_j \mid t_{<i}) \log p(t_j \mid t_{<i})$\;
}

$\mathcal{B} \gets \emptyset$\;
\ForEach{position $i$}{
    \If{$H(t_i) > \tau$ \textbf{or} $t_i$ is syntactic boundary}{
        Mark $i$ as unit boundary\;
    }
}
Partition $P'$ into semantic units $\mathcal{B} = \{b_1, \ldots, b_m\}$\;

$T \gets$ empty hierarchy with top unit $r$\;
$Q \gets$ queue initialized with $(r, 0, |\mathcal{B}|-1)$ 
\While{$Q$ is not empty}{
    $(v, s, e) \gets Q.\textsc{Dequeue}()$\;
    Partition semantic units $\mathcal{B}[s:e]$ by boundaries and indentation\;
    \ForEach{contiguous unit range $[s', e']$}{
        Create compositional sub-unit $u$ for units $\mathcal{B}[s':e']$\;
        Add $u$ as lower-level unit of $v$\;
        $Q.\textsc{Enqueue}((u, s', e'))$\;
    }
}

$\ourmetric(P) \gets 0$\;
\ForEach{unit $v \in T$}{
    $\ourmetric(P) \gets \ourmetric(P) + \alpha \cdot b(v) + (1-\alpha) \cdot d(v)$\;
}

\Return{$\ourmetric(P)$}
\end{algorithm}


\section{Extended Theoretical Analysis}
\label{sec:theory}

This section extends the theoretical analysis presented in Section~\ref{sec:theory-main}. We first introduce the formal setup and notation (\S~\ref{sec:formal-setup}), and then present additional theoretical results on the separation of \ourmetric from traditional complexity measures (\S~\ref{sec:separation}) and the justification for entropy-driven semantic segmentation (\S~\ref{sec:entropy-seg}).

\subsection{Formal Setup and Notation}
\label{sec:formal-setup}

Let $\mathcal{M}$ be an autoregressive language model processing a code token sequence $c = (t_1, t_2, \ldots, t_n)$. The predictive entropy at position $i$ is $H_i := H(t_i \mid t_{<i}) = -\sum_{v \in \mathcal{V}} P_{\mathcal{M}}(v \mid t_{<i}) \log P_{\mathcal{M}}(v \mid t_{<i})$, where $\mathcal{V}$ denotes the vocabulary.

A semantic unit $u = (t_l, \ldots, t_r)$ is a maximal contiguous subsequence satisfying: (i) for all $i \in (l, r)$, $H_i \leq \tau$ and $t_i$ is not a syntactic delimiter; and (ii) either $H_l > \tau$, $l = 1$, or $t_{l-1}$ is a syntactic delimiter. The collection of semantic units partitions the code sequence.

A semantic hierarchy $\mathcal{T} = (V, E)$ represents the compositional organization of program semantics, where each element $v \in V$ corresponds to a semantic unit and relations in $E$ encode syntactic composition between units. The compositional level of a unit $v$, denoted by $d(v)$, reflects its degree of nested composition, while the branching factor $b(v)$ quantifies the number of semantic units produced from the same compositional context. Let $L(v)$ denote the token length of unit $v$.
We define the unit-level conditional entropy as $H(v) := \sum_{i=l}^{r} H(t_i \mid t_{<i})$, where unit $v$ spans the token sequence $(t_l, \ldots, t_r)$.

\paragraph{Supporting Evidence for Assumptions.}
Assumption~\ref{assump:context} (Depth-Dependent Context Degradation) is supported by extensive empirical evidence. \citet{liu2024lost} demonstrate a ``lost in the middle'' phenomenon where language models exhibit significantly degraded performance when relevant information is positioned in the middle of long contexts, with a characteristic U-shaped performance curve favoring information at context boundaries. Furthermore, studies on attention patterns reveal that transformers allocate disproportionate attention to initial tokens (``attention sinks'')~\citep{xiao2024efficient}, and that effective context utilization degrades with distance even within the nominal context window. Recent work on intelligence degradation in long-context LLMs~\citep{wang2026intelligence} further confirms that model performance deteriorates progressively as context length increases beyond certain thresholds.

Assumption~\ref{assump:branch} (Branching Uncertainty) is motivated by recent findings on entropy-aware reasoning in language models. Studies on test-time compute~\citep{li2025entropy} show that high-entropy tokens correspond to critical decision points where prediction uncertainty is elevated, and that concentrating computation budget on such high-uncertainty moments improves reasoning quality. In the context of code, control-flow constructs (conditionals, loops) introduce structural ambiguity requiring the model to reason about multiple execution paths. Empirical studies confirm that LLM-generated code with deeper nesting and more comparisons exhibits higher error rates~\citep{sepidband2025enhancing}, and benchmarks specifically designed for nested structures~\citep{basu2025nestful,hu2025dynacode} consistently observe performance degradation as nesting complexity increases.

\subsection{Separation from Traditional Metrics}
\label{sec:separation}

We establish that \ourmetric captures semantic properties fundamentally different from cyclomatic complexity (CC).

\begin{proposition}[Separation from Cyclomatic Complexity]
\label{prop:cc-separation}
There exist code families $\{c_n^{\text{flat}}\}$ and $\{c_n^{\text{nested}}\}$ such that: (i) $\mathrm{CC}(c_n^{\text{flat}}) = \mathrm{CC}(c_n^{\text{nested}}) = n+1$; (ii) $|c_n^{\text{flat}}| = \Theta(|c_n^{\text{nested}}|)$; and (iii) $\ourmetric(c_n^{\text{flat}}) = \Theta(n)$ while $\ourmetric(c_n^{\text{nested}}) = \Theta(n^2)$.
\end{proposition}

\begin{proof}
\textbf{Construction of $c_n^{\text{flat}}$.} Consider $n$ sequential independent conditionals at the same nesting level:
\begin{verbatim}
if cond1: A1 else: B1
if cond2: A2 else: B2
...
if condn: An else: Bn
\end{verbatim}
This code has $n$ decision points, so $\text{CC} = n + 1$ by McCabe's formula.
For \ourmetric, the code decomposes into $1 + n + 2n = 3n + 1$ semantic units: one top-level unit, $n$ conditional units, and $2n$ branch bodies ($A_1, B_1, \ldots, A_n, B_n$), at compositional levels $1$, $2$, and $3$ respectively. Thus $\text{TotalCompLevel} = 1 + 2n + 6n = \Theta(n)$ and $\text{TotalBranch} = \Theta(n)$, yielding $\ourmetric(c_n^{\text{flat}}) = \Theta(n)$.

\textbf{Construction of $c_n^{\text{nested}}$.} Consider $n$ nested conditionals:
\begin{verbatim}
if cond1:
  if cond2:
    ...
      if condn: An else: Bn
    ...
  else: B2
else: B1
\end{verbatim}
This code also has $n$ decision points, so $\text{CC} = n + 1$.
For \ourmetric, the code decomposes into $2n + 1$ semantic units: $n$ conditional units at compositional levels $1, 2, \ldots, n$, $n-1$ else-branches ($B_1, \ldots, B_{n-1}$) each at level $i+1$ for the $i$-th conditional, and two terminal units ($A_n, B_n$) at level $n+1$. Thus $\text{TotalCompLevel} = \Theta(n^2)$ and $\text{TotalBranch} = \Theta(n)$, yielding $\ourmetric(c_n^{\text{nested}}) = \Theta(n^2)$.

\textbf{Code length.} The flat structure has length $O(n)$. The nested structure has length $O(n)$ for the code statements plus $O(n^2)$ for indentation whitespace. Since complexity metrics typically operate on token counts (not raw characters) and indentation tokens are often normalized, the effective lengths are $\Theta(n)$ in both cases.
\end{proof}

This separation has practical implications: code with identical cyclomatic complexity can exhibit vastly different LLM processing difficulty depending on hierarchical organization.
By Theorem~\ref{thm:entropy-accumulation}, deeper hierarchies induce greater entropy accumulation, so the $\Theta(n)$ vs.\ $\Theta(n^2)$ gap in \ourmetric reflects actual differences in LLM difficulty that CC cannot capture.

\begin{proposition}[Orthogonality to Code Length]
\label{prop:length-orthogonal}
For any code length $L$, there exist code samples $c_1$ and $c_2$ with $|c_1| = |c_2| = L$ such that $|\text{TotalCompLevel}(\mathcal{T}_1) - \text{TotalCompLevel}(\mathcal{T}_2)| = \Omega(L)$, where $\mathcal{T}_1$ and $\mathcal{T}_2$ denote the corresponding semantic hierarchies.
\end{proposition}

\begin{proof}
We construct two code samples of length $L$ tokens, each partitioned into $k = \Theta(\sqrt{L})$ semantic units of average length $\Theta(\sqrt{L})$.

\textbf{Flat structure ($c_1$).} All $k$ units reside at compositional level 1, yielding $\text{TotalCompLevel}(\mathcal{T}_1) = k = \Theta(\sqrt{L})$.

\textbf{Chain structure ($c_2$).} Units are arranged in a linear compositional chain with levels $1, 2, \ldots, k$, yielding $\text{TotalCompLevel}(\mathcal{T}_2) = \sum_{i=1}^{k} i = k(k+1)/2 = \Theta(L)$.

Thus, $|\text{TotalCompLevel}(\mathcal{T}_1) - \text{TotalCompLevel}(\mathcal{T}_2)| = \Theta(L) - \Theta(\sqrt{L}) = \Omega(L)$.
\end{proof}

These results explain why cyclomatic complexity loses predictive power under length control: CC counts decision points without encoding their hierarchical arrangement, and empirically $\text{CC}(c) \approx \beta_0 + \beta_1 |c|$, so length control eliminates the signal. In contrast, \ourmetric captures hierarchical organization—a dimension orthogonal to length.

\subsection{Entropy-driven Segmentation}
\label{sec:entropy-seg}

We provide theoretical justification for using entropy thresholds to identify semantic unit boundaries.

\begin{proposition}[Entropy Thresholding as Boundary Detection]
\label{prop:entropy-boundary}
Suppose the entropy sequence $\{H_i\}_{i=1}^n$ exhibits a bimodal structure: tokens within semantic units have entropy concentrated around a low level $\mu_L$, while boundary tokens have entropy concentrated around a high level $\mu_H > \mu_L$, each with noise bounded by $\eta$. That is, $H_i \in [\mu_L - \eta, \mu_L + \eta]$ for interior tokens and $H_i \in [\mu_H - \eta, \mu_H + \eta]$ for boundary tokens. If the gap satisfies $\mu_H - \mu_L > 2\eta$, then choosing threshold $\tau = (\mu_L + \mu_H)/2$ correctly classifies all tokens as interior or boundary.
\end{proposition}

\begin{proof}
For interior tokens: $H_i \leq \mu_L + \eta < (\mu_L + \mu_H)/2 = \tau$ (since $\mu_H - \mu_L > 2\eta$ implies $\mu_L + \eta < \mu_H - \eta$ and thus $\mu_L + \eta < (\mu_L + \mu_H)/2$).

For boundary tokens: $H_i \geq \mu_H - \eta > (\mu_L + \mu_H)/2 = \tau$ (by symmetric argument).

Thus, the threshold $\tau$ achieves perfect separation under the bimodal assumption.
\end{proof}

\begin{remark}
In practice, the threshold $\tau$ is selected empirically (e.g., using a quantile of the observed entropy distribution) rather than computed from unknown population parameters. Following~\citet{wang2025beyond}, we set $\tau$ at the 67th percentile of token entropies, which provides robust boundary detection across diverse code samples.
\end{remark}

This result formalizes the intuition that entropy spikes mark semantic boundaries. Empirical studies of token-level entropy in code~\citep{cooper2024perplexed,shi2025longcodezip} confirm that code entropy exhibits approximately bimodal behavior, with elevated values at control-flow transitions, function boundaries, and other structural discontinuities, justifying the threshold-based segmentation approach.

\section{Subgroup-based Correlation Analysis}
\label{sec:correlation-analysis}

To assess the correlation between code complexity metrics and LLM performance while mitigating sensitivity to individual samples, we adopt a subgroup-based correlation analysis. 
Specifically, samples are first ranked according to the feature or metric of interest and then partitioned into approximately ten equal-sized groups.

Within each group, we take the median feature (or metric) value as the representative statistic, which provides robustness against outliers and extreme cases in code structure or length. 
For task performance, we use the mean score within each group, as the median may degenerate to binary values (e.g., all 0 or all 1) in groups corresponding to particularly difficult or trivial instances, thereby obscuring performance variation.

We then compute the Spearman rank correlation coefficient $r$ between the group-level feature values and performance scores. 
To further mitigate the effects of grouping sensitivity and sample-level randomness, we repeat the analysis with the number of groups ranging from 9 to 11, and report the statistically significant result ($p < 0.05$) with the largest absolute correlation magnitude. 
This procedure yields more stable and reliable correlation estimates by attenuating the influence of individual sample variance and grouping artifacts.

\section {Semantics-Preserving Code Rewriting for LM-CC Reduction}
\label{sec:code-rewrite}

We conduct semantics-preserving code rewriting to directly reduce \ourmetric and examine its causal impact on model performance.
Specifically, we instruct an LLM to simplify each code example, and compute \ourmetric and cyclomatic complexity (CC) before and after rewriting. We retain only samples for which CC remains unchanged while \ourmetric decreases, ensuring that structural complexity is held constant. Semantic equivalence is verified using test cases: rewritten samples for code translation and execution reasoning are required to pass all tests, whereas those for program repair are required to fail the tests, consistent with task definitions.
We employ \textsc{DeepSeek-V3.2} for code rewriting with n=10 attempts and a temperature of 1. The prompt used for rewriting is provided below.

\begin{tcolorbox}[colback=gray!6, colframe=black!80, title=Prompt for Code Rewriting]
You are an expert in Python programming.
\\

\textbf{Instruction:}

- Improve the code's comprehensibility and readability without altering its original functionality.

- Do not modify the function's cyclomatic complexity (e.g., do not add or remove loops or conditional branches).

- Output only the simplified code inside 
\textasciigrave \textasciigrave \textasciigrave python \textasciigrave \textasciigrave \textasciigrave 
without any additional explanations.
\\

\textbf{Initial Code:} 

\{code\}
\\

\textbf{Simplified Code:}

\end{tcolorbox}

To exclude simple samples with limited simplification potential, we further filter the samples where the original \ourmetric ranked in the top 50\% (and in the top 60\% for the execution reasoning task).
This yields valid subsets of 58, 99, and 60 samples for program repair, code translation, and execution reasoning, respectively.

\section{Additional Experimental Results}
This section provides supplementary experimental analyses that further validate the effectiveness, robustness, and practical properties of \ourmetric.
We first report additional evaluation results using GPT-4o-mini and Qwen3.5-122B as task-performing models (\S~\ref{sec:main-results-sup}), followed by ablation studies examining the contribution of the hierarchical decomposition strategy and the weighting factor $\alpha$ (\S~\ref{sec:ablation}). Finally, we present a cost analysis of \ourmetric and discuss its computational efficiency in practical deployment settings (\S~\ref{sec:cost}).

\subsection{Supplementary Results on GPT-4o-mini and Qwen3.5-122B}
\label{sec:main-results-sup}

In addition to \textsc{DeepSeek-V3}, we further evaluate \ourmetric on two additional task-performing models, \textsc{GPT-4o-mini} and \textsc{Qwen3.5-122B}, using the same experimental setup as described in the main text.
The results, summarized in Table~\ref{tab:main_results_sup}, are highly consistent with those reported in Section~\ref{sec:experiments}, further supporting the effectiveness and generalizability of \ourmetric.

Specifically, \ourmetric achieves a perfect 16/16 significant correlation rate across all evaluated settings, consistently outperforming its two component dimensions, TotalCompLevel and TotalBranch.
Nevertheless, the effectiveness of different entropy estimation models varies across tasks and task-performing LLMs.
When \textsc{GPT-4o-mini} is adopted, \textsc{deepseek-coder-6.7b-base} and \textsc{Phi-3-mini-4k-instruct} achieve the strongest numerical performance across most tasks, whereas \textsc{CodeLlama-7b-hf} and \textsc{Qwen2.5-Coder-1.5B} perform competitively only on the code translation task.
In contrast, when \textsc{Qwen3.5-122B} is employed, nearly all entropy models obtain strong results on program repair and code translation, but exhibit relatively weak numerical performance on the execution reasoning task.

\begin{table*}[t]
\centering
\caption{Partial Spearman correlations ($r$) of complexity features and \ourmetric with pass@1 across code tasks using \textsc{GPT-4o-mini} and \textsc{Qwen3.5-122B}, while controlling for code length. ``--'' indicates non-significant correlations ($p \geq 0.05$).}
\label{tab:main_results_sup}
\resizebox{\textwidth}{!}{
\begin{tabular}{llccccccc}
\toprule
\multirow{2}{*}{\textbf{Entropy Model}} & \multirow{2}{*}{\textbf{Code Task}}&  \multicolumn{3}{c}{\textbf{Compositional Level}} & \multicolumn{3}{c}{\textbf{Branching Factor}} & \multirow{2}{*}{\textbf{\ourmetric}}\\
        \cmidrule(lr){3-5} \cmidrule(lr){6-8}
& & \textbf{MaxCompLevel}  & \textbf{AvgCompLevel}  & \textbf{TotalCompLevel} & \textbf{MaxBranch} & \textbf{AvgBranch} & \textbf{TotalBranch} & \\
\midrule

\rowcolor{gray!35} \multicolumn{9}{c}{\textbf{\textsc{GPT-4o-mini}}} \\

\midrule
\multirow{4}{*}{CodeLlama-7b}
& Program Repair      & - & - & \textbf{-0.93} & -  & -  & -0.83  & -0.85 \\
& Code Translation    & -0.67 & - & -0.93 & -0.92 & -0.80     & \textbf{-0.97} & \textbf{-0.97} \\
& Execution Reasoning & - & \textbf{-0.96} & -0.94 & - & -0.83     & -0.92 & -0.94 \\
& Code Summarization & -0.81 & -0.97 & -0.99 & \textbf{-1.0} & -  & -0.99 & -0.99 \\

\midrule
\multirow{4}{*}{Phi-3-mini-4k-instruct}
& Program Repair      & - & -0.75 & -0.88 & -  & -  & -0.69  & \textbf{-0.9} \\
& Code Translation    & - & - & -0.95 & -0.95 & -0.95     & \textbf{-0.98} & \textbf{-0.98} \\
& Execution Reasoning & - & -0.93 & -0.95 & -0.76 & -0.64     & -0.93 & \textbf{-0.98} \\
& Code Summarization & -0.95 & - & -0.99 & \textbf{-1.0} & -  & -0.99 & -0.99 \\

\midrule
\multirow{4}{*}{Qwen2.5-Coder-1.5B}
& Program Repair      & - & -0.80 & \textbf{-0.92} & -  & -  & -  & -0.80 \\
& Code Translation    & - & -0.75 & -0.94 & -0.91 & -0.88     & -0.96 & \textbf{-0.98} \\
& Execution Reasoning & -0.71 & \textbf{-0.90} & -0.88 & - & -     & -0.87 & -0.87 \\
& Code Summarization & -0.81 & -0.93 & -0.99 & \textbf{-1.0} & -  & -0.99 & -0.99 \\

\midrule
\multirow{4}{*}{deepseek-coder-6.7b-base}
& Program Repair      & - & -0.83 & -0.73 & -  & -  & -0.77  & \textbf{-0.91} \\
& Code Translation    & - & - & -0.94 & -0.94 & -0.88     & -0.93 & \textbf{-0.95} \\
& Execution Reasoning & -0.72 & -0.96 & -0.93 & -0.92 & -0.82     & \textbf{-0.98} & \textbf{-0.98} \\
& Code Summarization & -0.92 & -0.95 & -0.99 & \textbf{-1.0} & -  & -0.99 & -0.99 \\

\midrule
\rowcolor{gray!12} \multicolumn{2}{c}{\textit{Significant Correlation Ratio}}  & 7/16 & 11/16 & \textbf{16/16} & 10/16 & 7/16 & 15/16 & \textbf{16/16}  \\

\midrule
\rowcolor{gray!35} \multicolumn{9}{c}{\textbf{\textsc{Qwen3.5-122B}}} \\

\midrule
\multirow{4}{*}{CodeLlama-7b}
& Program Repair      & - & - & -0.82 & -  & -  & -0.85  & \textbf{-0.98} \\
& Code Translation    & -0.85 & - & -0.90 & -0.77 & -     & -0.85 & \textbf{-0.93} \\
& Execution Reasoning & - & -0.72 & -0.79 & - & \textbf{-0.82}     & -0.72 & -0.79 \\
& Code Summarization & -0.87  & -0.92 & -0.99 & \textbf{-1.0} & -  & -0.99 & -0.99 \\

\midrule
\multirow{4}{*}{Phi-3-mini-4k-instruct}
& Program Repair      & -0.68 & -0.75 & -0.95 & -0.79  & -  & -  & \textbf{-0.97} \\
& Code Translation    & - & - & \textbf{-0.88} & -0.82 & \textbf{-0.88}     & -0.84 & \textbf{-0.88} \\
& Execution Reasoning & -0.71 & - & -0.81 & \textbf{-0.89} & -0.80     & -0.80 & -0.76 \\
& Code Summarization & -0.95 & -0.74 & \textbf{-0.99} & -0.97 & -  & \textbf{-0.99} & \textbf{-0.99} \\

\midrule
\multirow{4}{*}{Qwen2.5-Coder-1.5B}
& Program Repair      & -0.85 & -0.68 & -0.92 & -  & -  & -  & \textbf{-0.98} \\
& Code Translation    & - & - & -0.76 & -0.69 & -     & -0.75 & \textbf{-0.77} \\
& Execution Reasoning & - & -0.71 & -0.80 & - & \textbf{-0.83}     & -0.82 & -0.80 \\
& Code Summarization & -0.9 & -0.87 & \textbf{-0.99} & -0.93 & -  & \textbf{-0.99} & \textbf{-0.99} \\

\midrule
\multirow{4}{*}{deepseek-coder-6.7b-base}
& Program Repair      & - & - & - & -  & -  & -0.90  & \textbf{-0.98} \\
& Code Translation    & - & - & \textbf{-0.88} & -0.74 & -     & -0.78 & -0.79 \\
& Execution Reasoning & - & \textbf{-0.84} & \textbf{-0.84} & - & -     & -0.80 & -0.71 \\
& Code Summarization & -0.95 & -0.88 & \textbf{-0.99} & -0.98 & -  & -0.99 & -0.96 \\

\midrule
\rowcolor{gray!12} \multicolumn{2}{c}{\textit{Significant Correlation Ratio}}  & 8/16 & 9/16 & 15/16 & 10/16 & 4/16 & 14/16 & \textbf{16/16} \\

\bottomrule
\end{tabular}
}
\end{table*}

\subsection{Ablation Study}
\label{sec:ablation}

\paragraph{Ablation on the Hierarchical Decomposition Strategy.}
The entropy-driven decomposition strategy in \ourmetric integrates token-level entropy with syntax-aware structural delimiters to construct a semantic compositional hierarchy over source code. To systematically evaluate the effectiveness and necessity of these two design elements, we conduct an ablation study by removing each one in turn.  The results are presented in Table~\ref{tab:ablation_results}.

\begin{table*}[htbp]
\centering
\caption{Ablation results on decomposition strategies for entropy and syntactic delimiter. ``--'' indicates non-significant correlations ($p \geq 0.05$). PR, CT, and ER stand for the three tasks of program repair, code translation, and execution reasoning, respectively.}
\label{tab:ablation_results}

\begin{tabular}{lccccccccc}
\toprule
\multirow{2}{*}{\textbf{Strategy}} & \multicolumn{3}{c}{\textbf{DeepSeek-V3}} & \multicolumn{3}{c}{\textbf{GPT-4o-mini}} & \multicolumn{3}{c}{\textbf{Qwen3.5-122b}} \\
\cmidrule(lr){2-4} \cmidrule(lr){5-7} \cmidrule(lr){8-10}
& \textbf{PR} & \textbf{CT} & \textbf{ER} & \textbf{PR} & \textbf{CT} & \textbf{ER} & \textbf{PR} & \textbf{CT} & \textbf{ER} \\
\midrule
Syntax-only & -0.88 & -0.78 & - & -0.85 & -0.93 & -0.78 & \textbf{-0.99} & -0.93 & -0.71 \\
Entropy-only & -0.72 & -0.95 & - & \textbf{-0.93} & \textbf{-0.99} & -0.90 & - & -0.84 & - \\
LM-CC & \textbf{-0.93} & \textbf{-0.97} & \textbf{-0.92} & -0.85 & -0.97 & \textbf{-0.94} & -0.98 & \textbf{-0.93} & \textbf{-0.79} \\

\bottomrule
\end{tabular}

\end{table*}

The results show that removing either entropy-guided boundaries (\textit{syntax-only}) or syntax-guided hierarchical structuring (\textit{entropy-only}) consistently reduces the stability of the correlation, and under certain experimental settings, the resulting coefficients fail to reach statistical significance. In contrast, \ourmetric, which jointly incorporates both design elements, achieves more stable and consistently competitive performance across configurations.

It should be noted that while depth and branching are well-established in traditional code complexity analysis, our key contribution is redefining them from an LLM-centric perspective. Traditional metrics use syntax-driven structures and reflect human-oriented complexity.
In contrast, LM-CC derives semantic units via entropy measurements and computes depth and branching scores over this LLM-aligned hierarchy, thereby capturing both the inherent uncertainty of LLM predictions and the compositional nonlinearity of code processing. Our ablation experiments confirm that syntax-based baselines with the same aggregation logic achieve inferior performance, proving that the improvement originates from LLM-driven hierarchical decomposition rather than the choice of complexity dimensions.

\paragraph{Ablation on the Weighting Factor $\alpha$.}

To investigate the impact of the weighting factor $\alpha$, we conduct an ablation study by computing \ourmetric under a range of $\alpha$ values.
The results are shown in Figure~\ref{fig:ablation}. 
As $\alpha$ increases, the effectiveness of \ourmetric remains relatively stable for code translation. For execution reasoning, it exhibits a mild rise followed by a gradual decline with limited variation, whereas for program repair it shows an increase followed by a sharper drop, indicating substantially higher sensitivity to $\alpha$.
Based on these observations, we set $\alpha = 0.8$, which yields consistently strong performance across all three tasks and strikes a favorable balance between the two structural features captured by \ourmetric.

\begin{figure}[t]
    \centering
    \includegraphics[width=0.5\columnwidth]{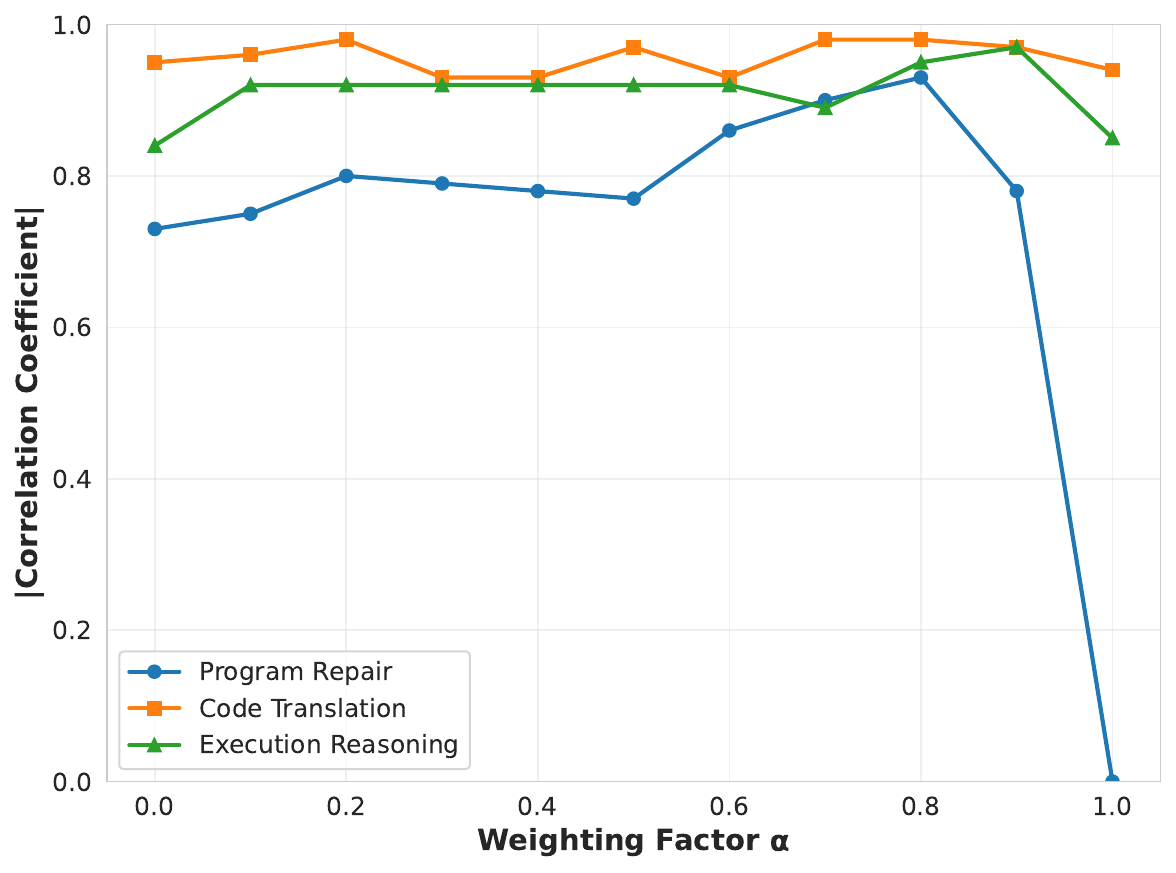}
    \caption{Ablation on the weighting factor $\alpha$ in \ourmetric. Performance peaks at intermediate $\alpha$ values, while hierarchy-only ($\alpha\!\to\!0$) and branching-only ($\alpha\!\to\!1$) configurations perform substantially worse.}
    \label{fig:ablation}
\end{figure}

\subsection{Cost Analysis}
\label{sec:cost}

\begin{table*}[t]
\centering
\caption{Comparison of average per-sample computation time between CC and LM-CC across different code tasks.}
\label{tab:cost_analysis}
\setlength{\tabcolsep}{12pt}
\begin{tabular}{lcc c}
\toprule
\multirow{2}{*}{\textbf{Task}} & \multirow{2}{*}{\textbf{Avg LoC}} & \multicolumn{2}{c}{\,\,\,\textbf{Avg Time/Sample}\,\,\,} \\
\cmidrule(lr){3-4}
& & \textbf{CC} & \textbf{LM-CC} \\
\midrule
Program Repair & 61.0 & 0.9ms & 0.7s \\
Code Translation & 20.8 & 0.4ms & 0.46s \\
Execution Reasoning & 8.4 & 0.1ms & 0.02s \\
Code Summarization & 16.9 & 0.2ms & 0.44s \\
\bottomrule
\end{tabular}

\end{table*}
LM-CC computes token-level entropy via a single forward pass, eliminating the need for autoregressive decoding. As a result, its computational complexity is linear in the input sequence length (O(n) tokens). Table~\ref{tab:cost_analysis} summarizes the average per-sample computation time for both Cyclomatic Complexity (CC) (0.1–0.9 ms/sample) and LM-CC (0.02–0.7 s/sample) across all four evaluation tasks. 
Although LM-CC incurs higher latency than lightweight static analysis metrics, it remains operationally feasible. It supports full parallelization via mini-batching and delivers far more expressive semantic representations, making it an attractive choice for large-scale empirical studies.

\section{Implications of \ourmetric}
\label{sec:implications}

The proposed \ourmetric enables several potential applications in LLM-centered software engineering and AI-assisted programming:

\textbf{LLM-Aware Code Evaluation and Benchmark Design}:
\ourmetric allows benchmarks and datasets to be stratified according to difficulty as perceived by large language models. By explicitly controlling for LLM-perceived complexity, it enables more faithful assessment of model capabilities, supports fairer cross-model comparisons, and facilitates finer-grained analysis of structural failure modes that are not captured by code length or traditional complexity measures.

\textbf{LLM-Guided Code Refactoring and Simplification}:
\ourmetric provides a principled signal for semantics-preserving code transformations that target reductions in LLM-perceived complexity. In contrast to traditional refactoring, which primarily optimizes for human readability or maintainability, model-aware refactoring restructures semantic compositional patterns to improve LLM performance without altering program behavior. This is particularly relevant for tasks such as program repair, code translation, code execution reasoning, and code understanding.

\textbf{Adaptive Reasoning and Tool Invocation}:
\ourmetric can serve as a control signal for dynamically selecting reasoning strategies in LLM-based systems. Code segments with higher perceived complexity may trigger structured prompting, explicit step-by-step reasoning, or external tool invocation (e.g., execution or symbolic analysis), whereas simpler code can be handled with lightweight generation, enabling cost- and performance-aware orchestration.

\textbf{Training Data Curation and Curriculum Learning}:
By quantifying code difficulty from the model’s perspective, \ourmetric can guide the construction of training curricula that better align with LLM-specific processing challenges. Ordering or weighting samples by model-aware complexity may improve training efficiency, robustness, and generalization compared to curricula based solely on syntactic properties or code length.

\textbf{Failure Analysis and Model Diagnosis}:
By explicitly capturing structural sources of difficulty such as deep composition and high branching intensity of semantic units, \ourmetric supports systematic diagnosis of LLM failures on code tasks. It helps explain performance degradation on complex programs—even when they are short—offering interpretable insights into how program semantic composition interacts with model behavior.

\end{document}